\documentclass[a4paper,pdfa,cleveref, autoref,USenglish,thm-restate]{lipics-v2021}
\nolinenumbers
\hideLIPIcs  

\usepackage{tabularx}
\newcommand{\problemdef}[3]{
	\begin{center}
		\begin{minipage}{0.95\textwidth}
			\noindent
			#1
			\vspace{5pt}\\
			\setlength{\tabcolsep}{3pt}
			\begin{tabularx}{\textwidth}{@{}lX@{}}
				\textbf{Input:}& #2 \\
				\textbf{Question:}& #3
			\end{tabularx}
		\end{minipage}
	\end{center}
}

\title{Directed Temporal Tree Realization for Periodic Public Transport: Easy and Hard Cases} 

\titlerunning{Directed Temporal Tree Realization for Periodic Public Transport} 

\author{Julia Meusel}{Martin Luther University Halle-Wittenberg, Germany}{julia.meusel@informatik.uni-halle.de}{https://orcid.org/0009-0001-2880-1390}{}{} 

\author{Matthias Müller-Hannemann}{Martin Luther University Halle-Wittenberg, Germany}{matthias.mueller-hannemann@informatik.uni-halle.de}{https://orcid.org/0000-0001-6976-0006}{}

\author{Klaus Reinhardt}{Martin Luther University Halle-Wittenberg, Germany}{klaus.reinhardt@informatik.uni-halle.de}{https://orcid.org/0009-0002-7002-4051}{}

\authorrunning{J. Meusel, M. Müller-Hannemann, K. Reinhardt} 

\Copyright{Julia Meusel, Matthias Müller-Hannemann, Klaus Reinhardt} 


\category{} 

\relatedversiondetails{A 5-page version has appeared as a brief announcement in SAND 2025}{https://doi.org/10.4230/LIPIcs.SAND.2025.21} 
\relatedversiondetails{Full version to appear in ATMOS 2025}{https://doi.org/10.4230/OASIcs.ATMOS.2025.8}

\begin{document}

\maketitle

\begin{abstract}

We study the complexity of the \emph{directed periodic temporal graph realization} problem.  
This work is motivated by the design of periodic schedules in public transport with constraints on the quality of service. Namely, we require that the fastest path between (important) pairs of vertices is upper bounded by a specified maximum duration, encoded in an upper distance matrix $D$. While previous work has considered the undirected version of the problem, the application in public transport schedule design requires the flexibility to assign different departure times to the two directions of an edge.  
A problem instance can only be feasible if all values of the distance matrix are at least shortest path distances. However, the task of realizing exact fastest path distances in a periodic temporal graph is often too restrictive. Therefore, we introduce a minimum slack parameter $k$ that describes a lower bound on the maximum allowed waiting time on each path.
We concentrate on tree topologies and provide a full characterization of the complexity landscape with respect to the period $\Delta$ and the minimum slack parameter~$k$, showing a sharp threshold between NP-complete cases and cases which are always realizable. We also provide  hardness results for the special case of period $\Delta = 2$ for general directed and undirected graphs. 
\end{abstract}
\section{Introduction}
\label{sec:Introduction}

Designing periodic schedules for public transport is notoriously difficult. Periodic schedules are desirable for several practical and operational reasons. They are easier to memorize, and travelers can better plan their journeys  if services run at regular intervals. Periodicity also enables coordinated connections between different lines at central hubs and simplifies crew and vehicle scheduling. In this paper, we consider the quality of service provided by a periodic schedule from a passenger's perspective. Specifically, we address the question of whether it is possible to design a periodic schedule for a given network that guarantees travel time bounds between important pairs of stops. We model this as a graph realization problem.

Graph realization problems are a central area of research that has been studied extensively since the 1960s for undirected~\cite{erdos1960graphs, Hakimi1965DistanceMO,Havel1955} and directed graphs~\cite{Anstee_1982,CHEN1966406,fulkerson1960zero}.
Given a set of constraints, the objective is to find a graph that satisfies them, or to decide that no such graph exists. 
Restrictions on degrees~\cite{AYOUB1970303,erdos1960graphs,Havel1955}, distances between vertices~\cite{BARNOY2024114810,Hakimi1965DistanceMO,tamura1993realization},  eccentricities~\cite{Lesniak1975}, and connectivity~\cite{Edmonds1964,KleitmanWang1976} have been studied in detail. 
Recently, the study of realization problems on temporal graphs was started by Klobas et al.~\cite{TGR_arxive,TGR_sand}.
Temporal graphs are graphs that have a fixed set of vertices and a set of edges that changes over time. 
Each edge of a static graph is assigned a set of timestamps at which it is active. In a periodic temporal graph, the set of timestamps is repeated periodically for all edges. 
Informally, a temporal path is a sequence of consecutive edges in the underlying static graph and corresponding increasing timestamps at which they are active.
Klobas et al.\ and Mertzios et al.\ examined restrictions on the travel time between pairs of vertices in periodic temporal graphs: upper bounds as well as exact values~\cite{TGR_arxive,TGR_sand,ubTTR_pre}. 

This leads to several variants of the \textsc{Temporal Graph Realization} problem (\textsc{TGR}) that have been studied:
In the periodic \textsc{TGR}, all timestamps are repeated periodically. The \emph{simple} periodic \textsc{TGR} allows only one timestamp per period, unlike the \emph{multi-label} periodic \textsc{TGR}.
The constraints on fastest travel times can be either exact values or upper bounds.
If the given underlying graph is a tree, we call the problem \textsc{Temporal Tree Realization} (\textsc{TTR}).
In addition to communication networks such as social networks or satellite links, Klobas et al.\ mention transportation networks as examples of potential applications for the \textsc{TTR}~\cite{TGR_arxive,TGR_sand}.
However, unlike the flow of information in satellite links, transportation networks typically do not carry passengers on an edge in both directions at the same time.
Therefore, we examine the directed case as a natural generalization. 
Since most public transportation lines run in both directions, we will mostly consider bidirected graphs.
In tree structures this is necessary to ensure that every vertex is reachable from every other vertex. 

With an application in public transport in mind, 
the input graph models the infrastructure network where vertices correspond to locations of stops (or stations) and edges connect neighboring stops served by a bus, tram, train or the like. In public transport, typical values for the period $\Delta$ are 5, 10, 15 or 20 minutes in urban transport, and 30, 60 or 120 minutes in long-distance train networks. The timestamp of an edge $(u,v)$ can be interpreted as the departure time of some vehicle at $u$.
In a practical setting, traveling along an edge $e$ requires $\ell_e \in \mathbb{N}$ time. An edge of length~$\ell$ can be equivalently replaced by a simple path of $\ell$ edges of unit length as shown in~\cref{fig:MinEx}. While such a transformation blows up the graph size, it does not change the complexity of the problem (since hardness results are established even for unit length graphs). For simplicity, we will therefore consider only unit-length graphs.

\begin{figure}[tb]
	\begin{subfigure}[t]{0.23\textwidth}
		\centering
		\includegraphics[scale=0.62]{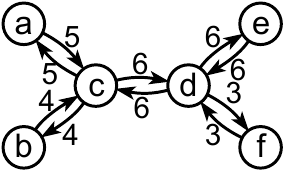}
		\caption{An infrastructure network $G$ labeled with distances (in minutes) between vertices (stops).}
		\label{fig:MinEx1}
	\end{subfigure}
	\hfill
	\begin{subfigure}[t]{0.74\textwidth}
		\centering
		\includegraphics[scale=0.62]{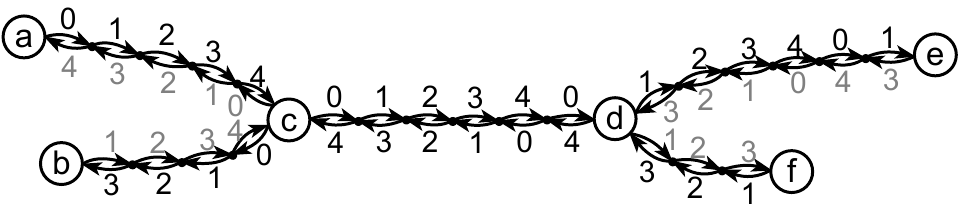}
		\caption{A graph in which phantom stops have been inserted so that all edges have unit length. Here, the labels define a schedule where each value denotes a departure time with respect to a global period of $\Delta=5$. The resulting graph with its labeling forms a periodic temporal graph.}
		\label{fig:MinEx2}
	\end{subfigure}
	\caption{Example: The infrastructure network $G=(V,E)$ with edge lengths (left) can be modelled with unit length edges by inserting phantom stops between already existing stops (right). No upper bounds are set for travel times from and to these phantom stops. Suppose we have a period of $\Delta=5$ and the following upper bounds: $D_{a,e} = 17$, $D_{a,b} = 9$, and $D_{f,b} = 13$. Assuming, without loss of generality, that the label of the edge connecting $a$ to the first phantom stop is assigned the timestamp 0, the given upper bounds on the travel times enforce the black edge labels.
 A temporal graph with edge labels as specified respects these upper bounds. The timestamps of a fastest temporal path from $f$ to $e$ may be, for example, $1,2,3,6,7,8,9,10,11$, giving it a duration of~11, whereas the static distance of $f$ and $e$ is only $9$.  On this fastest path, one would wait at vertex~$d$ for two timesteps before traversing the next edge with label $1=6\bmod 5$  at time $6$. Therefore, adding an upper bound $D_{f,e}=9$ would make the instance infeasible. }
	\label{fig:MinEx}
\end{figure}

Informally, given a directed, strongly connected graph as the underlying static graph as well as a period~$\Delta$ and upper bounds on the travel time between each pair of vertices, the objective is to compute a (single) periodic timestamp for each directed edge such that a fastest temporal path between any two vertices does not exceed the given upper bound. 
The upper bounds between pairs of vertices can be interpreted as a guaranteed quality of service that must be met. The bounds are specified by a  matrix $D$ of integers or $\infty$. The latter means that we do not impose any restriction on the fastest path between the corresponding vertices.
Motivated by  transportation networks with fixed traffic lines, we consider only one global period  instead of allowing different periods for all edges.
Since waiting times are unavoidable and natural, for example, due to transfer times at transit hubs, we assume that there is a small amount of waiting time allowed on all shortest paths and introduce a minimum slack parameter~$k$ that specifies how much waiting time is at least acceptable on all shortest routes. This parameter gives us some leeway to work with, as waiting is necessary even on very small instances, such as the one shown in \cref{fig:MinEx}.

Travelers on a public transport network typically have to change trains along the way. In practice, a minimum transfer time must be planned for each change.
For simplicity, we consider only the case where all transfer times are 0. However, one could easily extend the model by imposing non-trivial minimum transfer times at each stop. More precisely, one could add constraints specifying that the difference between the arrival time at a stop and the departure time on a leaving edge must be at least a given stop-specific constant (the minimum transfer time at this stop).

\subparagraph*{Related work.}
Klobas et al.\ show that the \textsc{Simple Periodic TGR} with exact given fastest travel times is \textsc{NP}-hard even for a small constant period $\Delta \geq 3$~\cite[Theorem 3]{TGR_arxive,TGR_sand}.
However, if the underlying static graph is a tree, the problem is solvable in polynomial time~\cite[Theorem 27]{TGR_arxive,TGR_sand}.
It is fixed-parameter tractable (FPT) with respect to the feedback edge number  of the underlying graph but W[1]-hard when parameterized by the feedback vertex number of the underlying graph~\cite[Theorem 29 and Theorem 4]{TGR_arxive,TGR_sand}.

While the \textsc{Simple Periodic Temporal Graph Realization} problem with exact given fastest travel times is solvable in polynomial time on trees, Mertzios et al.\ showed that the problem given upper bounds on the fastest travel times is \textsc{NP}-hard even if the underlying static graph is a tree or even a star ~\cite[Theorem 5]{ubTTR_pre}. 
This still holds for a constant period of $\Delta=2$ and when the input tree $G$ has a constant diameter or a constant maximum degree~\cite[Theorem 5]{ubTTR_pre}. 
However, it is FPT with respect to the number of leaves in the input tree $G$~\cite[Theorem 19]{ubTTR_pre}.

While Klobas et al.\ and Mertzios et al.\ only allow one timestamp per edge, Erlebach et al.\ consider several timestamps per edge~\cite{paraTGR_sand}. 
They examine both the periodic and the non-periodic variant with exact given fastest travel times.
Among other results, they show that the  \textsc{Multi-Label Periodic Temporal Graph Realization} problem is \textsc{NP}-hard, even if the underlying static graph is a star for any number $\ell\geq 5$ of labels per edge.  
All these models have in common that labels (periodic timestamps) are assigned to edges. 

Important related problem versions assign periodic labels to vertices.
The \textsc{Periodic Event Scheduling Problem} (PESP), introduced by Serafini and Ukovich in 1989~\cite{PESP}, is widely used to schedule reoccurring events in public transport. Here the input is a so-called \emph{event-activity network} $N=(V,A)$, a period $\Delta$, and time windows $[\ell_a,u_a]$ for each activity $a\in A$. The set $V$ models \emph{events} (think of an arrival or departure of a vehicle at a stop) and the set $A \subset V \times V$ so-called \emph{activities}. Activities model driving between neighboring stops, dwelling at a stop, transfers between different vehicles or safety constraints (minimal headways). 
In PESP, one seeks a periodic timestamp $\pi_v \in  \{0, 1, \dots, \Delta-1\}$ for each event $v \in V$ such that $(\pi_w - \pi_v) \bmod \Delta \in [\ell_{(v,w)},u_{(v,w)}]$ for all $(v,w) \in A$. 
The time window $[\ell_a,u_a]$ of an activity models lower and upper bounds on the difference of the timestamps between the corresponding events. In other words, the difference $u_a - \ell_a$ bounds the slack (waiting time) which can be introduced on activity $a$. 
In stark contrast to our model, these restrictions are local constraints between adjacent events, not global constraints between arbitrary events as considered in this paper. The PESP is known to be \textsc{NP}-complete for fixed $\Delta \geq 3$ \cite{PESP,Odijk1994}, but efficiently solvable for $\Delta = 2$~\cite[page 87]{Peeters2003}. Deciding the feasibility of a PESP instance is \textsc{NP}-hard even when the treewidth is 2, the branchwidth is 2, or the carvingwidth is 3~\cite{LindnerReisch2022}.


\begin{figure}[t]
	\begin{minipage}{0.99\textwidth}
		\centering
		\includegraphics[scale=0.85]{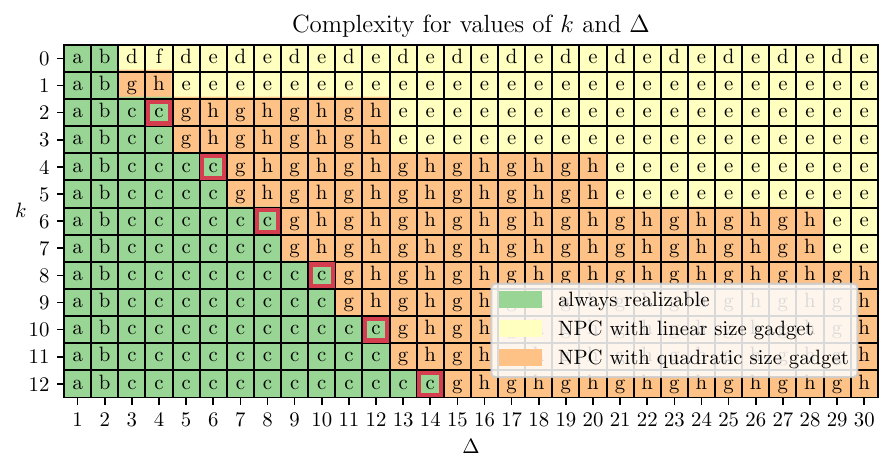}
		\caption{Complexity of the \textsc{Directed Upper-Bounded Periodic Temporal Tree Realization Problem} for different values of the minimum slack parameter $k$ and period $\Delta$ (\textsc{$\Delta$-$k$-DiTTR}).
		The labeling of the boxes indicates which proof was employed to achieve the respective result,
		where \cref{thm:RedToDir} is used for all \textsc{NP}-complete cases: 
			(a)~\cref{thm:Delta1},
			(b)~\cref{thm:Always2},
			(c)~\cref{thm:Always},
			(d)~\cref{thm:GadgetY}, 
			(e)~\cref{thm:GadgetX}, 
			(f)~\cref{thm:Delta4}, 
			(g)~\cref{thm:GadgetKu}, and 
			(h)~\cref{thm:GadgetKg}.
	Cases in which the undirected problem version is NP-complete but all instances of the directed version are realizable are outlined in red (see \cref{thm:RedToDir}).   
	}
		\label{fig:hardness}
	\end{minipage}	
\end{figure}

\subparagraph*{Our contribution.} In this paper, we investigate the complexity of the  \textsc{Directed Upper-Bounded Periodic Temporal Graph Realization Problem} (\textsc{DiTGR}) and the \textsc{Directed Upper-Bounded Periodic Temporal Tree Realization Problem} (\textsc{DiTTR}).
Our main results are as follows:
\begin{itemize}
	\item We provide an efficiently checkable necessary and sufficient condition for feasibility in DiTTR when all pairs of vertices must be realized on shortest paths without waiting time. The basic insight is that the solvability for such instances depends on the distance between branching vertices of the given graph topology.
	\item We then introduce the parameter $k$, which specifies the minimum waiting time to be allowed on each shortest path (the \emph{slack}), i.e., for all pairs of vertices $u$ and $v$ we require that the duration bound $D_{u,v}$ is at least the distance of $u$ and $v$ in the underlying static graph plus the constant $k$.
	We fully characterize the complexity of the problem for bidirected tree topologies in terms of period $\Delta$ and  slack parameter $k$. For each possible combination of parameters $\Delta$ and $k$, we either prove that the problem is easily solvable or hard, see \cref{fig:hardness}.
	\item 
For $\Delta \geq 3$, we give a simpler proof that the undirected version of the \textsc{Upper-Bounded Periodic Temporal Tree Realization} problem is \textsc{NP}-complete, even if the underlying static graph is a star and thus has a constant diameter\footnote{Independently of us, the authors of \cite{ubTTR_pre} included essentially the same result in version 2 of their arxiv paper, while it was raised as an open question in version 1.}.   
	 \item 
	 We also investigate in more detail the special case of a given period $\Delta = 2$. This turns out to be \textsc{NP}-complete in general (for both directed and undirected\footnotemark[\value{footnote}] graphs), 
but can be solved efficiently on directed bipartite graphs and graphs in which all shortest paths are unique. 
\end{itemize}


\subparagraph*{Organization of the paper.}
In Section~\ref{sec:problem}, we start with a formal problem definition. Then, in Section~\ref{sec:feasible} we will first characterize feasible instances by providing a necessary and sufficient condition and then present all cases where instances with a bidirected tree topology can be easily solved. 
We will also discuss in more detail the special case of a given period $\Delta = 2$. 
In Section~\ref{sec:hardness}, we give our hardness results for the remaining instances with a directed bidirected tree topology. Finally, we conclude in Section~\ref{sec:conclusions} with a summary and suggestions for future work.

\section{Formal Problem Definition}\label{sec:problem}
As we mainly consider the problem for directed graphs, all following definitions deal with directed temporal graphs. 
Both undirected edges and directed arcs are referred to as \emph{edges}.
The definitions for undirected temporal graphs are analogous.
For temporal graphs and periodic temporal graphs, we follow the notation of Klobas et al.~\cite{TGR_arxive,TGR_sand}:
\begin{definition}\label{def:tempGraph}
	A \emph{temporal graph} is a pair $(G,\Lambda)$, where $G=(V,E)$ is the underlying (static) graph and $\Lambda: E \rightarrow 2^{\mathbb{N}_0}$ is a function, that assigns a set of discrete timestamps to each edge. 
\end{definition}
\begin{definition}\label{def:periodictempGraph}
	A \emph{$\Delta$-periodic} temporal graph is a triple $(G=(V,E),\lambda: E \to \{0,1,\ldots,\Delta-1\},\Delta)$ which denotes the temporal graph $(G, \Lambda)$ where $\forall e \in E: \Lambda(e)=\{\lambda(e)+i\cdot \Delta\mid i \in \mathbb{N}_0\}$. 
\end{definition}
Informally, a temporal path is a sequence that denotes consecutive edges on a path in the underlying static graph and the times at which they are traversed. 
No vertex can be visited more than once. 
	Recent literature distinguishes between the strict and non-strict version. Throughout the paper we only consider strict paths:
The timestamps have to be strictly increasing.
Formally, we can define a temporal path as follows:
\begin{definition}\label{def:tempPath}
	A \emph{temporal $s$-$z$-path} of length $\ell$ in a directed temporal graph $(G,\Lambda)$ is a sequence $P=(v_{i-1},v_i,t_i)_{i=1}^\ell $ for which the following holds:
	\begin{itemize}
		\item $v_0=s \land v_{\ell}=z$
		\item $\forall i,j \in \{0,\mathellipsis, \ell\}, i\neq j: v_i \neq v_j$
		\item $\forall i \in \{1,\mathellipsis, \ell\}: (v_{i-1},v_{i})\in E $
		\item $\forall i \in \{1,\mathellipsis, \ell\}: t_i \in \Lambda((v_{i-1},v_{i}))$
		\item $\forall i \in \{2,\mathellipsis, \ell\}: t_{i-1} < t_{i}$
	\end{itemize}
\end{definition}

The traversal of an edge requires one time unit. The temporal $s$-$z$-path \emph{starts} or \emph{begins} at vertex $s$ at time $t_1$, it \emph{reaches} or \emph{arrives} at vertex $z$ at time $t_\ell+1$.

\begin{definition}\label{def:tempPathDur}
	The \emph{duration} $d(P)$ of a temporal path $P=(v_{i-1},v_i,t_i)_{i=1}^\ell $ is defined as $d(P)=t_\ell+1-t_1$.
\end{definition}

Let $\hat{d}(u,v)$ be the static \emph{distance} of $u$ and $v$ in the underlying static graph.
The undirected and directed problem versions can now be stated as follows.

\problemdef{\textsc{Periodic Upper-Bounded Temporal Graph Realization (TGR)}}
{An undirected, connected graph $G=(V,E)$ with $V=\{v_1,v_2,\ldots,v_n\}$, an $n \times n$  matrix $D$ of positive integers where $\forall u,v\in V: D_{u,v}\geq \hat{d}(u,v)$, and a positive integer $\Delta$.}
{Does there exist a $\Delta$-periodic labeling $\lambda: E \rightarrow \{0,1,\ldots,\Delta-1\}$ such that, 
	for every~$i,j$, the duration of a fastest temporal path from $v_i$ to $v_j$ in the $\Delta$-periodic temporal graph $(G,\lambda,\Delta)$ is \emph{at most} $D_{i,j}$?}

The restriction of \textsc{TGR} where the given graph is a tree is called \textsc{TTR}. To generalize the undirected problem version to the directed one, we restrict the considered graphs to the simplest case: we only consider directed graphs obtained by replacing each undirected edge  with two antiparallel directed edges. We call graphs that are created this way bidirected.

\problemdef{\textsc{Periodic Upper-Bounded Temporal Directed Graph Realization (DiTGR)}}
{A directed, strongly connected graph $G=(V,E)$ with $V=\{v_1,v_2,\ldots,v_n\}$, an $n \times n$ matrix $D$ of positive integers where $\forall u,v\in V: D_{u,v}\geq \hat{d}(u,v)$, and a positive integer $\Delta$.}
{Does there exist a $\Delta$-periodic labeling $\lambda: E \rightarrow \{0,1,\ldots,\Delta-1\}$ such that, 
	for every~$i,j$, the duration of a fastest temporal path from $v_i$ to $v_j$ in the $\Delta$-periodic temporal graph $(G,\lambda,\Delta)$ is \emph{at most} $D_{i,j}$?}

 The restriction of \textsc{DiTGR} to inputs of such graphs derived from trees by adding two directed edges $(u,v)$ and $(v,u)$ for every undirected edge $\{u,v\}$ is called \textsc{DiTTR}.
 This means that for any pair of vertices $(u,v)$ there is exactly one path from $u$ to $v$ in the underlying static graph. 
Furthermore, we only consider instances where $\forall u,v\in V: D_{u,v}\geq \hat{d}(u,v)$, because all other instances cannot be realized anyway. 
The \emph{duration} of a fastest temporal path from $u$ to $v$ depending on $\lambda$ is denoted by $d_\lambda(u,v)$.
We simply write $d(u,v)$ whenever $\lambda$ is clear from the context.  
For brevity, we write $\lambda(u,v)$ instead of $\lambda((u,v))$.
The \emph{waiting time} at vertex $v_i$ on a path $P=(v_{i-1},v_i,t_i)_{i=1}^\ell$  is $t_{i+1}-t_i -1$ for $1 \leq  i <\ell$.
The waiting time on a path $P$ is the sum of the waiting time at all its vertices.
In a bidirected tree it is equal to the difference $d(v_0,v_\ell)-\hat{d}(v_0,v_\ell)$ on a fastest path.
In \cref{fig:MinEx2}, the static distance $\hat{d}(f,e)$ between the vertices $f$ and $e$ is 9, whereas the duration of the fastest temporal path between them is $d(f,e)=11$, with a waiting time of $2=6-3-1$ only at vertex $d$.
For a vertex $v$, we denote by $\delta^+(v)$ its static outdegree, by $\delta^-(v)$ its static indegree, and by $\delta(v) = \delta^+(v) + \delta^-(v)$ its (total) static degree; $N(v)$ refers to the set of its neighbors.

For any pair of vertices $(v,w)$ the duration of a fastest temporal path $P=(v_{i-1},v_i,t_i)_{i=1}^\ell$ can be at most $d(P)\leq (\hat{d}(v,w)-1)\cdot \Delta +1$.
This bound is achieved if $\lambda(v_{i-1},v_i)$ is equal for all $i \in \{1, \mathellipsis, \ell\}$.
Therefore, any value of $D_{v,w}$ with $D_{v,w}\geq (\hat{d}(v,w)-1)\cdot \Delta+1$ is no real restriction.
We write any such value simply as  $D_{v,w}=\infty$ or omit it entirely.

The slack parameter $k$ is an implicit parameter of $D$: $D$ is restricted by $\forall v,w \in V: D_{v,w}\geq \hat{d}(v,w)+k$. For trees, this means that $k$ is the minimum permitted waiting time on any path, since paths are unique. We call the versions of \textsc{TTR} and \textsc{DiTTR} where the parameters  $\Delta$ and $k$ are fixed \textsc{$\Delta$-$k$-TTR} respectively \textsc{$\Delta$-$k$-DiTTR}.

\section{Characterization of Feasible Instances}\label{sec:feasible}

\subsection{A Necessary and Sufficient Condition for Feasibility in \textsc{DiTTR}}
Before we look at the hardness results for the \textsc{DiTTR} problem, let us consider a special case: There is no waiting time allowed on any shortest path. As the underlying static graph is a tree, all paths are shortest paths and no waiting is allowed on any path. Therefore, the given travel times are not just upper bounds but exact values. Since there is only one path between every pair of vertices, this means $\forall u,v \in V: D_{u,v}=\hat{d}(u,v)$.
 For undirected graphs, this special case is also covered by the result of Klobas et al.\ for the exact \textsc{TTR}: it is decidable in polynomial time whether the instance is realizable~\cite[Theorem 27]{TGR_arxive,TGR_sand}.
We provide a simple characterization of the realizability of an instance of \textsc{DiTTR}.  
To do this, we introduce so-called branching vertices. We will call a vertex with a degree of at least 6 \emph{branching vertex}.

\begin{observation}\label{obs:exact}
	An instance of \textsc{DiTTR} with $\forall u,v \in V: D_{u,v}=\hat{d}(u,v)$ is realizable exactly when the distance of any pair of branching vertices is a multiple of $\Delta/2$.
\end{observation}
\begin{proof}
First we observe that a vertex with a static degree of at least 6 corresponds to a vertex with a static degree of at least 3 in the underlying undirected tree. 
As can easily be seen from the following argument, all branching vertices have fixed arrival and departure times, i.e.\ all incoming edges of a branching vertex have the same label, as well as all outgoing edges.
Let $v$ be a branching vertex and let for some neighbor $w$ of $v$ w.l.o.g. $\lambda(w,v)=\Delta-1$. 
For all outgoing edges $(v,x)$ of $v$ excluding $(v,w)$ this means $\lambda(v,x)=0$ as no waiting is allowed.
This in turn requires $\lambda(x, v)=\Delta-1$ for the remaining incoming edges $(x, v)$. The latter implies $\lambda(v,w)=0$.
Let $y$ and $z$ be two branching vertices with distance $\hat{d}(y,z)$ and let $t(y)$, $t(z)$  be the timestamps of their incoming edges.
Then the timestamps of their outgoing edges have to be $(t(y) +1) \bmod \Delta$ and $(t(z)+1)\bmod \Delta$, respectively.
Let $P=(y=v_0,v_1,t_1)\mathellipsis(v_{\ell-1},z=v_\ell,t_\ell)$ be a fastest temporal $y$-$z$-path and its length $\ell$. 
By \cref{def:tempPathDur}, the duration of $P$ is $d(P)=t_\ell-t_1+1$. We note that for any feasible solution $d(P)=d(y,z) = \hat{d}(y,z)$. 
To not exceed $D_{y,z}=\hat{d}(y,z)$ the following must hold: 
$\exists i\in \mathbb{N}_{0}: i \cdot \Delta+ t(z) = t_\ell=d(P)+t_1-1=\hat{d}(y,z) +(t(y)+1)  -1$.
By definition and due to symmetry, this means:
$t(z) \equiv\hat{d}(y,z)+ t(y) \pmod{\Delta}$ and $t(y) \equiv \hat{d}(y,z)+t(z) \pmod{\Delta}$.
Therefore, the following must also apply: $0 \equiv 2 \cdot \hat{d}(y,z) \pmod{\Delta}$.

If the distance of any pair of branching vertices is a multiple of $\Delta/2$, we can start at an arbitrary branching vertex r and set its arrival time to $\Delta-1$. 
Without waiting, the departure time of this vertex is $0$. 
Then we can assign labels iteratively as enforced by already labeled adjacent edges as specified in Algorithm~\ref{lis:RealAlg}. (However, we emphasize that we have to choose the root as a branching vertex here.)
As all branching vertices have a distance of a multiple of $\Delta/2$, this procedure assigns them fixed arrival and departure times with no waiting.
Hence, with such a labeling there is a temporal path between every pair of vertices without waiting.
If there are no branching vertices at all, we can assign increasing labels independently for both directions of the path.
\end{proof}

This necessary and sufficient condition implies immediately a linear time algorithm for these instances. This does not contradict our observation that the problem becomes NP-complete if waiting times are allowed, i.e., $D_{u,v}\geq\hat{d}(u,v)$. This is even true for the slack parameter $k=0$, since $k$ is a only lower bound for the allowed waiting time.

\subsection{Efficiently Solvable Cases}

In this section we demonstrate how to realize all instances with $\Delta \leq k+1$ for odd $\Delta$ and with $\Delta \leq k+2$ for even $\Delta$. 
Just for completeness, we start with the trivial case $\Delta = 1$.

\begin{theorem}\label{thm:Delta1}
	For $\Delta=1$ all instances of \textsc{TGR} and \textsc{DiTGR} are feasible.
\end{theorem}
\begin{proof}
Obviously, all labels are forced to have the same value $\lambda = 0$. Since the period $\Delta$ is~1, there is no waiting time at all at any vertex. So every shortest path in the underlying static graph can be realized as a fastest temporal path with duration equal to its length, which means that all instances are realizable. 
\end{proof}

\begin{lstlisting}[mathescape=true,caption={Realizing instances that can always be realized}, label=lis:RealAlg, captionpos=t, abovecaptionskip=-\medskipamount, float]
Choose an arbitrary vertex $r$ as root. 
For each edge $e=(a,b)$, do:
	if $e$ points away from the root, 
		set $\smash{\lambda(a,b)=\hat{d}(r,a) \bmod \Delta}$
	if $e$ points to the root, 
		set $\smash{\lambda(a,b)=(\Delta -\hat{d}(r,a)) \bmod \Delta}$
\end{lstlisting}
\begin{figure}[t]
	\begin{subfigure}[b]{0.33\textwidth}
		\centering
		\includegraphics[scale=0.7]{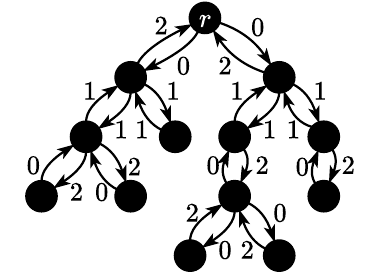}
		\caption{Realization for $k+1\geq \Delta = 3$}
		\label{fig:sol1}
	\end{subfigure}
	\hfill
	\begin{subfigure}[b]{0.33\textwidth}
		\centering
		\includegraphics[scale=0.7]{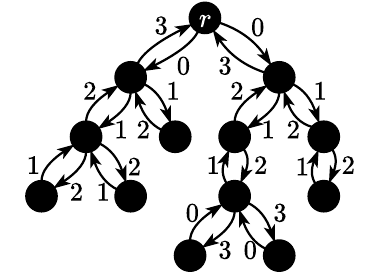}
		\caption{Realization for $k+2 \geq \Delta = 4$ }
		\label{fig:sol2}
	\end{subfigure}
	\caption{Example: Two realizations with waiting times at most $\Delta-1$ respectively $\Delta-2$ as constructed by Algorithm \ref{lis:RealAlg}}
	\label{fig:sols}
\end{figure}

\begin{theorem}\label{thm:Always}
	All instances of \textsc{DiTTR} with $\Delta \leq k+1$ for odd $\Delta$ and with $\Delta \leq k+2$ for even $\Delta$ are feasible.
\end{theorem}
\begin{proof}
All instances with $\Delta \leq k+1$ for odd $\Delta$ and with $\Delta \leq k+2$ for even $\Delta$ can be realized by a simple algorithm detailed in Algorithm \ref{lis:RealAlg}.
First, choose an arbitrary vertex $r$ as the root. 
Then, traverse the tree and, depending on its distance from the root and whether it points to or away from the root, assign a label to each edge as specified in Algorithm~\ref{lis:RealAlg}.
Two realizations computed by this algorithm can be seen in  \cref{fig:sols}.
	There are three cases for the direction of a path between two vertices: 
	\begin{enumerate}
		\item The path runs entirely towards the root.
		\item The path leads  strictly away from the root.
		\item The path first heads towards the root and then changes direction away from it.
	\end{enumerate}
	In the first two cases there is no delay at all. 
	For the last case, it is easy to see that waiting time can only occur when changing direction and therefore only once per path.
	Thus, the maximum waiting time is at most $\Delta-1$ on any path.
	If $\Delta$ is even, the waiting time is at most $\Delta-2$, because by construction  incoming and outgoing edges of every vertex are assigned values of different parity. Hence, there cannot be two identical labels in a row on a path.
\end{proof}

\begin{corollary}\label{thm:Always2}
	All instances of \textsc{DiTTR} with $\Delta =2$ are feasible.
\end{corollary}

This follows from \cref{thm:Always} with $\Delta=2$ and $k=0$ and can easily be generalized to directed bipartite graphs.

\begin{corollary}\label{thm:Always3}
	All instances of \textsc{DiTGR} where the input is restricted to a directed bipartite graph $G=(U,V,E)$ with $E \subseteq (U\times V) \cup (V \times U)$ and period $\Delta =2$ are feasible.
\end{corollary}
\begin{proof}
Given any directed, bipartite graph $G=(U,V,E)$ with $E \subseteq (U\times V) \cup (V \times U)$ and $\Delta =2$, set $\lambda(e_0)=0$ for all $e_0 \in E \cap (U\times V)$ and $\lambda(e_1)=1$ for all $e_1\in E \cap (V\times U)$.
Since every path in $G$ alternates between vertices in $U$ and vertices in $V$, there is no waiting time.
\end{proof}

A graph where all shortest paths are unique is called \emph{geodetic}~\cite[p. 105]{GBV-196749549} or \emph{min-unique}~\cite{RA00}.
\begin{restatable}{thm}{cycles}\label{thm:OddCycle}
    Instances of \textsc{DiTGR} with period $\Delta =2$ which are restricted to a directed geodetic graph $G=(V,E)$ with  $D_{u,v}\in\{\hat{d}(u,v),\infty\}$ for all $u,v \in V$ can be decided in polynomial time.
\end{restatable}
\begin{proof}
	We show this by giving a reduction to \textsc{2--Vertex-Coloring} as shown in \cref{fig:1Tree}.
	Note that as the graph is geodetic there is by definition a unique shortest path in the underlying static graph for every pair of vertices $(v,w)$~\cite[p. 105]{GBV-196749549}.
	Given an instance  $(G=(V,E), D, \Delta)$, let $\mathtt{P}\subseteq 2^E$ be the set of all edge sets of paths in $G$ and let $\mathtt{P_e}\subseteq\mathtt{P}\setminus\{\emptyset\}$ be the set of edge sets of all shortest paths where no waiting time is allowed. 
	We construct an undirected auxiliary graph $G'=(E,E')$ with 
	$E'=\{\{(x,y),(y,z)\} \mid \exists P \in \mathtt{P_e}: (x,y),(y,z) \in P\}$.
	This means that the labeling has to be chosen such that $\lambda(x,y)=(1-\lambda(y,z)) \bmod \Delta$ for any $\{(x,y),(y,z)\} \in E'$ since waiting at $y$ on the way from $x$ to $z$ is not allowed.
	As $\Delta=2$, this is true if and only if $\lambda(x,y) \neq \lambda(y,z) \text{ for all } \{(x,y),(y,z)\} \in E'$.
	Thus, given a feasible coloring $C: E \rightarrow \{0,1\}$ for $G'$, we can set $\lambda(x,y)=C((x,y)) \text{ for all } (x,y)\in E$. 
	Therefore, the instance is feasible if and only if $G'$ is bipartite which can be decided in polynomial time.
\end{proof}
\begin{figure}[tb]
	\begin{subfigure}[t]{0.4\textwidth}
		\centering
		\includegraphics[width=\linewidth]{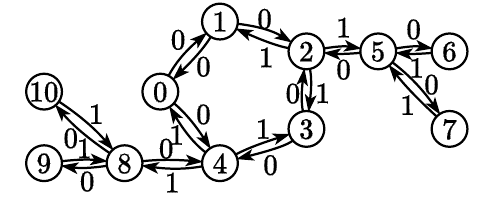}
		\caption{A graph $G$ with a feasible labeling for the case $D'_{0,10}=3$.}
		\label{fig:1Treea}
	\end{subfigure}
	\hfill
	\begin{subfigure}[t]{0.27\textwidth}
		\centering
		\includegraphics[width=.95\linewidth]{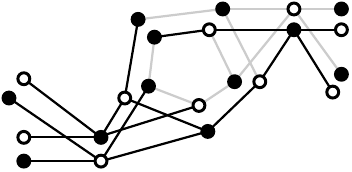}
		\caption{The bicolored graph $G'$ corresponding to the labeling of $G$ in \cref{fig:1Treea} with $D'_{0,10}=3$.}
		\label{fig:1Treeb}
	\end{subfigure}
	\hfill
	\begin{subfigure}[t]{0.27\textwidth}
		\centering
		\includegraphics[width=.95\linewidth]{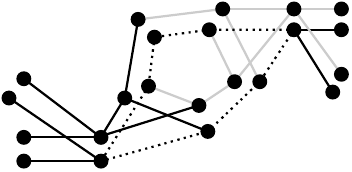}
		\caption{The graph $G''$  with $D''_{1,10}=4$ is not bipartite. The dotted lines indicate an cycle of odd length.}
		\label{fig:1Treec}
	\end{subfigure}
	\caption{Example: A graph $G=(V,E)$ containing one bidirected cycle of odd length and two attached trees. The auxiliary graphs $G'=(E,E')$ and $G''=(E,E'')$ for two different matrices $D'$ and $D''$ with $D'_{3,0}=D''_{3,0}=2$, $D'_{6,0}=D''_{6,0}=4$, $D'_{7,9}=D''_{7,9}=6$, $D'_{9,1}=D''_{9,1}=4$, $D'_{10,3}=D''_{10,3}=3$ and $D'_{0,10}=3$, $D''_{1,10}=4$. Black and dotted edges are derived by the specific $D'$ or $D''$, whereas gray edges correspond to potential additional constraints for general $D$. The graph $G_1'$ in \cref{fig:1Treeb} is bicolored. The colors of the vertices correspond to the feasible labeling shown in \cref{fig:1Treea}. The graph $G''$ is not bipartite, the instance with $D''$ is therefore infeasible.}
	\label{fig:1Tree}
\end{figure}
In sharp contrast, we will show in the following section, that \textsc{DiTGR} is NP-complete for general graphs even for the special case that $\Delta =2$.

\section{Hardness Results}\label{sec:hardness}

In this section, we present several hardness results. First, we provide a simple alternative proof of the \textsc{NP}-completeness of \textsc{TTR} for $\Delta \geq 3$. Second, we extend the previously known hard cases by showing that \textsc{TGR} and \textsc{DiTGR} are hard for $\Delta = 2$.
Third, we present hardness results for all remaining cases of values for  $k$ and  $\Delta$ for \textsc{DiTTR}.

\subsection{Warm-up: Simplified Proof for NP-completeness of \textsc{TTR}}
Mertzios et al.~\cite{ubTTR_pre} showed that  \textsc{TTR} is \textsc{NP}-complete, even with constant $\Delta$ and if the input graph has constant diameter or constant maximum degree. 
We provide a simpler proof for the first case by reducing from \textsc{$\Delta$-Coloring}. In particular, we also show that \textsc{TTR} is \textsc{NP}-complete even if the input graph is a star, answering an open question of~\cite[version~1]{ubTTR_pre}~\footnote{Independently of us, the authors meanwhile answered their question and obtained essentially the same results in version 2 of their paper.}.

Given an instance $G=(V,E)$ for \textsc{$\Delta$-Coloring} we create a new star-shaped graph $G'=(V',E')$ that has a single new vertex $u$ at its center and all the vertices of $G$ surrounding~$u$.
\begin{align*}
		 V'&=V \cup \{u\}\\
	 E'&=\{\{u,v\}\mid v \in V\}\\
	D_{v_1,v_2} &=
	\begin{cases}
		\Delta & \{v_1,v_2\}\in E \\
		\infty & \text{else}
	\end{cases} &\text{for all } v_1, v_2 \in V
\end{align*}
Setting $D_{v_1,v_2}=\Delta$ ensures $\lambda(u,v_1) \neq \lambda(u,v_2)$, because if they were equal, the path $P=(v_1,u,\lambda(u,v_1)) (u,v_2,\lambda(u,v_2)+\Delta)$ would be a fastest temporal $v_1$-$v_2$-path. 
Therefore, the path's duration would be $d(v_1,v_2)=\lambda(u,v_2)+\Delta-\lambda(u,v_1)+1=\Delta+1>D_{v_1,v_2}$.
Thus, given a feasible solution for \textsc{TTR}, for all pairs of vertices $(v_1,v_2)$ that are neighbors in $G$, the following holds: $\lambda(u,v_1) \neq \lambda(u,v_2)$. 
Therefore, we can set the color of a vertex $v$ to $\lambda(u,v)$ to get a valid solution for \textsc{$\Delta$-Coloring}. Conversely, every valid \textsc{$\Delta$-Coloring} immediately implies a valid solution for \textsc{TTR} if we assign the color of vertex $v$ to $\lambda(u,v)$ for all $v \in V$.
The resulting graph has a constant diameter of 2, but no constant maximum degree, and the constraints $D$ are symmetric.
This proof works for any period $\Delta >2$.
Note that the reduction from vertex coloring implies strong NP-completeness.

\subsection{NP-completeness of \textsc{DiTGR} and \textsc{TGR} for $\Delta = 2$}

\begin{theorem}\label{thm:2DTGR_NPC}
	\textsc{DiTGR} is NP-complete even for the special case where period $\Delta =2$, the edges $E$ are bidirected, the constraints $D$ are symmetric
	and the only odd cycle in $G$ is a triangle.
\end{theorem}
\begin{figure}[t]
	\begin{minipage}{0.99\textwidth}
		\centering
		\includegraphics[scale=0.75]{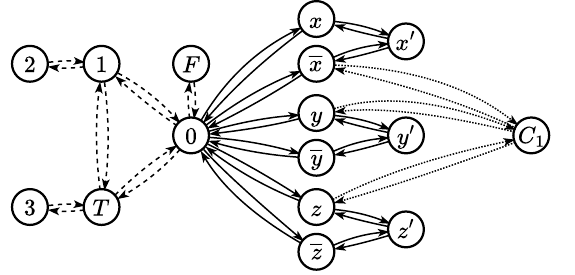}
		\caption{Construction with gadgets for the variables $x$, $y$, $z$ and for the clause $C_1=( \overline{x}\lor y \lor z)$. The dashed edges indicate the basic structure. The edges for the variable gadgets are solid, and those for the clauses are dotted.}
		\label{fig:Cgadget}
	\end{minipage}
\end{figure}

\begin{proof}
We reduce \textsc{3-SAT} to \textsc{DiTGR} as follows: Given a 3-CNF formula $\Phi=\{C_1, \mathellipsis, C_m\}$ with a set of variables $X$, construct a graph $G=(V,E)$ 
with
\begin{align*}
	V& =\{T,F,0,1,2,3\}\cup \{x,\overline{x}, x'|x\in X\}\cup \Phi \\
	E&=\{(2,1),(1,2),(3,T),(T,3),(1,T),(T,1),(1,0),(0,1),(T,0),(0,T),(F,0),(0,F)\} \\
	&\cup \{(0,x),(x,0),(0,\overline{x}),(\overline{x},0)|x \in X\} \\
	&\cup
	\{(x,x'),(x',x),(\overline{x},x'),(x',\overline{x})|x \in X\}\\
	&\cup \{(l,C),(C,l)|l\in \{x,\overline{x}|x\in X\} \land C\in \Phi \land l\in C\}
\end{align*} 
 as shown in \cref{fig:Cgadget} and set the upper bounds $D$ as follows:
\begin{align*}
	&D_{2,F}=D_{F,2}=D_{3,2}=D_{2,3}=D_{3,F}=D_{F,3}=3\\
	&D_{T,x'}=D_{x',T}=D_{F,x'}=D_{x',F}=3 &\text{ for all } x\in X\\
	&D_{x,\overline{x}}=D_{\overline{x},x}=2 &\text{ for all } x\in X\\ &D_{T,C}=D_{C,T}=3 &\text{ for all } C\in \Phi
\end{align*}

All other values of $D$ are set to $\infty$.
For each pair of vertices with a finite distance bound, this enforces that we have to realize a labeling without waiting on some shortest path.
Let without loss of generality $\lambda(T,1)=0$. Then the upper bounds $D$ between the vertices $2$, $3$ and $F$ enforce that $\lambda(T,1)=\lambda(1,T)=\lambda(T,0)=\lambda(0,T)=\lambda(1,0)=\lambda(0,1)=0$ and $\lambda(3,T)=\lambda(T,3)=\lambda(2,1)=\lambda(1,2)=\lambda(F,0)=\lambda(0,F)=1$.

For each variable $x\in X$, there exist two possible shortest paths from $0$ and therefore from~$T$ respectively $F$ to $x'$: one over $x$ and one over $\overline{x}$. 
Since the paths from $T$ and the paths from $F$ start with different labels ($0$ respectively $1$), they also have to continue with different labels to avoid waiting.
This means
$\lambda(0,x)=1-\lambda(0,\overline{x})=1-\lambda(x,x')=\lambda(\overline{x},x')$.
The same holds for the opposite direction: $\lambda(x,0)=1-\lambda(\overline{x},0)=1-\lambda(x',x)=\lambda(x',\overline{x})$.
Then $D_{x,\overline{x}}=D_{\overline{x},x}=2$ enforces that the labels in both directions have to be the same, i.e. $\lambda(0,x)=\lambda(x,0)$:
Going without waiting from $x$ to $\overline{x}$ over 0 requires  $\lambda(x,0)=1-\lambda(0,\overline{x})=\lambda(0,x)$.
Going over $x'$ requires  $\lambda(x,x')= 1-\lambda(x',\overline{x})=\lambda(x',x)$ which results in the same labeling.

For every clause $C=\{l_1,l_2,l_3\}$ there are three shortest paths from $T$ to $C$: each one leads over one of the literals $l_1$, $l_2$ and $l_3$ in C.
A fastest path can only have a duration of 3 as enforced by the upper bounds, if the edge $(0,l)$ to the literal $l$ has the label $1$. 
Therefore at least one of the edges $(0,l_1)$, $(0,l_2)$, $(0,l_3)$ has to have the label $1$.

An assignment $a: X \mapsto \{0,1\}$ corresponds to the labeling of the edges from $0$ to the variables,
i.e. $\lambda(0,x)=\lambda(x,0)=a(x)$ for all $x \in X$.
If $\Phi \in$ \textsc{3-SAT} and $a$ is a satisfying assignment, then we
extend $\lambda$ to $\lambda(l,C)=\lambda(C,l)= 1-\lambda(0,l)$ for every clause $C\in \Phi$ which is satified by the literal $l$, which leads to no waiting on the path between $T$ and $C$
and thus $(G,D,2) \in$ \textsc{DiTGR}.
If $(G,D,2) \in$ \textsc{DiTGR} with the labeling $\lambda$, then the corresponding assignment $a$ must satisfy $\Phi \in$ \textsc{3-SAT}.
\end{proof}
From the proof, we can therefore conclude:

\begin{corollary}\label{thm:2TGR_NPC}
	\textsc{TGR} is NP-complete even for the special case that $\Delta =2$.
\end{corollary}

\subsection{NP-complete Cases with Linear Size Gadgets}

In this and the following subsection we present hardness results for all other cases, i.e.\ for all pairs of values $(k,\Delta)$ as shown in \cref{fig:hardness}. 
In all cases, the basic idea is to construct gadgets to enforce that for some specific edge $(v_1,v_2)$ the timestamps for both directions have the same value, i.e. $\lambda(v_1,v_2)=\lambda(v_2,v_1)$.
However, every value of $\lambda$ is possible, it is not restricted by the gadget. 
In fact, every solution implies $\Delta-1$ further symmetrical solutions in which all values are shifted by some value $x < \Delta$.
By enforcing this for every single edge of an instance $I$, we can reduce the undirected \textsc{TTR} to \textsc{DiTTR}.
All gadgets considered in this paper have a size polynomial in $\Delta$ and $k$. We distinguish between cases where we found linear-size gadgets and cases where we found quadratic-size gadgets.

\begin{lemma}\label{lemma:gadgetsRedu}
	If there is a polynomial size gadget which is a realizable instance of \textsc{$\Delta$-$k$-DiTTR} and which enforces $\lambda(v_1,v_2)=\lambda(v_2,v_1)$ for some specific pair of vertices $(v_1,v_2)$, we can reduce \textsc{$\Delta$-$k$-TTR} to \textsc{$\Delta$-$k$-DiTTR}.
\end{lemma}

\begin{proof}
	Let $I=(G=(V,E),D,\Delta)$ be an instance of \textsc{TTR} and  $(\hat{G},\hat{D},\Delta)$ be a gadget enforcing $\lambda(a,b)=\lambda(b,a)$.
	First we create a directed graph $G'$ by replacing every undirected edge of $G$ by two antiparallel directed edges.
	Then we enforce $\lambda(v,w)=\lambda(w,v)$ for every edge $e=\{u,v\}\in E$ by inserting a copy of the gadget into $G'$, identifying $(u,v)$ with $(a,b)$ and $(v,u)$ with $(b,a)$, and setting $D'$ consistently with $D$ and $\hat{D}$.
	Since the resulting graph is a tree ($G$ and the gadget $G'$ are trees that are connected by one common edge) and any two copies of the gadget share at most one common vertex the construction does not produce any shortcuts. 
	
	Thus, every valid solution $\lambda'$ for \textsc{DiTTR} with respect to $G'$ and $D'$ immediately implies a valid solution $\lambda$ for \textsc{TTR} if we set $\lambda(\{u,v\})=\lambda'((u,v))=\lambda'((v,u))$.
	Conversely, we can construct a valid solution for \textsc{DiTTR} from a valid solution for \textsc{TTR} by setting $\lambda'((u,v))=\lambda'((v,u))=\lambda(\{u,v\})$.
	Since the gadget is feasible, there exists a solution where the timestamp of the edges $(a,b)$ and $(b,a)$ is 0. We can then set the labels in the copy of the gadget for each edge $\{u,v\}$ to a solution that is shifted modulo $\Delta$ by $\lambda(\{u,v\})$.
	Since $\Delta$ and $k$ are fixed, this construction is possible in polynomial time if the time to compute the gadget is a function of only $\Delta$ and $k$. 
\end{proof}

\begin{restatable}{lem}{GadgetY}\label{thm:GadgetY}
	We can construct a gadget $G_{\Delta,0}$ that enforces $\lambda(v_1,v_2)=\lambda(v_2,v_1)$ for some pair of vertices $(v_1,v_2)$ for odd period $\Delta$ and minimum slack $k=0$.
\end{restatable}
\begin{proof}
	We construct this gadget as follows:
	\begin{align*} 
		V&=\{1,\mathellipsis,4+\lfloor \frac{\Delta}{2}\rfloor \}\\
		E&=\{(1,3),(3,1),(2,3),(3,2)\}\\
		&\cup \{(i,i+1)| i \in \{3,\mathellipsis, 3+\lfloor \frac{\Delta}{2}\rfloor\}\} 
		\cup \{(i+1,i)| i \in \{3,\mathellipsis, 3+\lfloor \frac{\Delta}{2}\rfloor\}\}\\
		D_{1,2}&=D_{2,1}=2\\
		D_{1,4+\lfloor \frac{\Delta}{2}\rfloor}&=D_{4+\lfloor \frac{\Delta}{2}\rfloor,1}=D_{2,4+\lfloor \frac{\Delta}{2}\rfloor}=D_{4+\lfloor \frac{\Delta}{2}\rfloor,2}=2+\lfloor \frac{\Delta}{2}\rfloor
	\end{align*}
	All other values of $D$ are set to $\infty$.
	As no waiting time is allowed, vertex 3 has a fixed arrival/departure time (see proof for \cref*{obs:exact}). 
	Let $t(3)$ be the timestamps of its incoming edges.
	For the sake of convenience, let us assume that $t(3) = \Delta-1$ and that the departure time is $0$. 
	This implies $\lambda(3,4)=0$ and therefore $\lambda(3+\lfloor \frac{\Delta}{2}\rfloor,4+\lfloor \frac{\Delta}{2}\rfloor) = \lfloor \frac{\Delta}{2}\rfloor$.
	Symmetrically, the following is also true: $\lambda(4,3)=\Delta-1$ and $\lambda(4+\lfloor \frac{\Delta}{2}\rfloor,3+\lfloor \frac{\Delta}{2}\rfloor) = \Delta-1-\lfloor \frac{\Delta}{2}\rfloor$.
	As $\Delta$ is odd this means $\lambda(4+\lfloor \frac{\Delta}{2}\rfloor,3+\lfloor \frac{\Delta}{2}\rfloor)=\lambda(3+\lfloor \frac{\Delta}{2}\rfloor,4+\lfloor \frac{\Delta}{2}\rfloor)= \lfloor \frac{\Delta}{2}\rfloor$. Therefore we can enforce $\lambda(v_1,v_2)=\lambda(v_2,v_1)$ for the vertices $v_1=3+\lfloor \frac{\Delta}{2}\rfloor$ and $v_2=4+\lfloor \frac{\Delta}{2}\rfloor$.
\end{proof}
For $ \Delta=3$ the gadget is shown in \cref{fig:gadget1}. In this gadget, $\lambda(5,4)=\lambda(4+\lfloor \frac{\Delta}{2}\rfloor,3+\lfloor \frac{\Delta}{2}\rfloor)=\lambda(3+\lfloor \frac{\Delta}{2}\rfloor,4+\lfloor \frac{\Delta}{2}\rfloor)=\lambda(4,5)$ is enforced.
This gadget has a size linear in $\Delta$ and only six values in $D$ that are not infinity.
Furthermore, $D$ is symmetric.

\begin{figure}[bt]
	\begin{minipage}[b]{0.3\textwidth}
		\centering
		\includegraphics[scale=0.7]{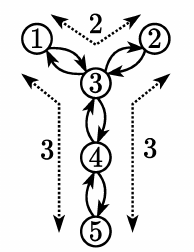}
		\caption{Gadget for $\Delta=3$ and $k=0$ that enforces $\lambda(4,5)=\lambda(5,4)$. The dotted lines indicate  $D_{v,w} <\infty$.}
		\label{fig:gadget1}
	\end{minipage}
	\hfill
	\begin{minipage}[b]{0.3\textwidth}
		\centering
		\includegraphics[scale=0.7]{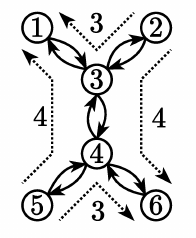}
		\caption{Gadget for $\Delta=5$ and $k=1$ that enforces $\lambda(3,4)=\lambda(4,3)$. The dotted lines indicate  $D_{v,w} <\infty$.}
		\label{fig:gadget2}
	\end{minipage}
	\hfill
	\begin{minipage}[b]{0.3\textwidth}
		\centering
		\includegraphics[scale=0.7]{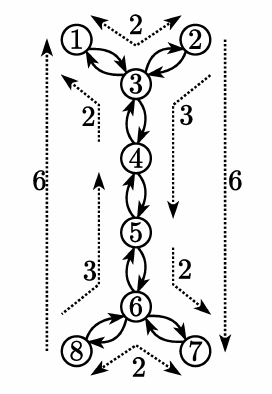}
		\caption{Gadget for $\Delta=4$ and $k=0$ that enforces $\lambda(4,5)=\lambda(5,4)$. The dotted lines indicate $D_{v,w} <\infty$.}
		\label{fig:gadget3}
	\end{minipage}
\end{figure}

\begin{restatable}{lem}{GadgetX}\label{thm:GadgetX}
	There is a gadget $G_{\Delta,k}$ that even with the limitation $D_{u,v}\geq \hat{d}(u,v)+k$ enforces $\lambda(v_1,v_2)=\lambda(v_2,v_1)$ for some pair of vertices $(v_1,v_2)$  as long as $\Delta \geq 4\cdot k +1 \land k\geq 1$ and minimum slack $k$ is odd.
	For any even values of $k$  with $\Delta \geq 4\cdot k +5$ we can simply use the gadget for~$k+1$.
\end{restatable}

\begin{proof}
	The following gadget enforces $\lambda(3+\lfloor \frac{k}{2}\rfloor,3+\lceil \frac{k}{2}\rceil)=\lambda(3+\lceil \frac{k}{2}\rceil,3+\lfloor \frac{k}{2}\rfloor)$:
	\begin{align*}
		V&=\{1,\mathellipsis,k+5 \}\\
		E&=\{(k+3,k+4),(k+4,k+3),(k+3,k+5),(k+5,k+3)\}
		\\ &\cup \{(1,3),(3,1),(2,3),(3,2)\}
		\\ &\cup \{(i,i+1)| i \in \{3,\mathellipsis, k+2\}\} 
		\\ &\cup \{(i+1,i)| i \in \{3,\mathellipsis, k+2\}\}
		\\D_{2,1}&=D_{k+4,k+5}=k+2
		\\D_{2,k+5}&=D_{k+4,1}=2\cdot k +2
	\end{align*}
	All other values of $D$ are set to $\infty$.
	Let w.l.o.g. $\lambda(k+4,k+3)=0$.
	We show that there is only one solution  by looking at possible values of $\lambda$ for the edges from and to leaves.
	We observe $\hat{d}(k+4,1)=\hat{d}(2,k+5)=k+2$.
	As there is no waiting required but at most a waiting time of $k$ allowed on the path from vertex $k+4$ to vertex $1$, the following must hold: $\lambda(3,1) \in  \{k+1, \mathellipsis, 2\cdot k +1 \}$ (recall $\Delta \geq 4\cdot k +1$).
	That means in turn that $\lambda(2,3)$ is in the range $\{0,1, \dots, 2\cdot k\}$. The corner case $\lambda(2,3) = 0$
	occurs when $\lambda(3,1) = k+1$ and with maximum waiting time.
	Conversely, with $\lambda(3,1)=2\cdot k +1$ and with no waiting time, $\lambda(2,3)$ is at most $2\cdot k $.
	
	In the same way we can conclude that $\lambda(k+3,k+5) \in \{k+1, \mathellipsis, 4\cdot k , (4\cdot k +1) \bmod \Delta  \}$.
	Because $\Delta \geq 4\cdot k +1$ only the last item may be affected by the modulo operator and would be assigned the value 0 in that case.
	Suppose now that $\lambda(k+3,k+5) \in \{k+2, \mathellipsis, 4\cdot k, (4\cdot k +1) \bmod \Delta \}$, i.e. any other possible value than $k+1$.
	Then $d(k+4, k+5)\geq k+2 -0 +1 =k+3 > D_{k+4,k+5}=k+2$ and thus, the solution cannot be valid for these values of $\lambda(k+3,k+5)$.
	Therefore $\lambda(k+3,k+5) =k+1$ which means that there is no waiting time at all on the paths from $k+4$ to $1$ and from $2$ to $k+5$ but a waiting time of $k$ on the paths from $2$ to $1$ and from $k+4$ to $k+5$.
	Thus, $\lambda(3+\lfloor \frac{k}{2}\rfloor,3+\lceil \frac{k}{2}\rceil)=\lceil \frac{k}{2}\rceil=\lambda(3+\lceil \frac{k}{2}\rceil,3+\lfloor \frac{k}{2}\rfloor)$.
	Therefore we can enforce $\lambda(v_1,v_2)=\lambda(v_2,v_1)$ for the vertices $v_1=3+\lceil \frac{k}{2}\rceil$ and $v_2=3+\lfloor \frac{k}{2}\rfloor$.
\end{proof}
\cref{fig:gadget2} shows the gadget for $\Delta=5$ and $k=1$.
This gadget also has a size linear in $k$ and therefore also in $\Delta$ and a constant amount of values in $D$ that are not infinity. 
The constraints $D$ are not symmetric; however, setting them as such doesn't affect the proof: For all feasible solutions $\lambda(3+\lfloor \frac{k}{2}\rfloor,3+\lceil \frac{k}{2}\rceil)=\lambda(3+\lceil \frac{k}{2}\rceil,3+\lfloor \frac{k}{2}\rfloor)$ holds and there is still at least one feasible solution.

\begin{restatable}{lem}{GadgetFour}\label{thm:Delta4}
	There is a gadget for $\Delta=4$ and $k=0$ that enforces $\lambda(v_1,v_2)=\lambda(v_2,v_1)$ for some pair of vertices $(v_1,v_2)$.
\end{restatable}
The rough idea is that we need a break in symmetry. So we start with two copies of a known gadget (see \cref{thm:GadgetY}) that we merge at the edges for which we can enforce equality of the labels.
The resulting gadget would be infeasible, so we relax the constraints $D$ slightly, making them asymmetric, and thereby obtain a feasible gadget of constant size with the claimed property.
\begin{proof}
	The following gadget enforces $\lambda(4,5)=\lambda(5,4)$.
	It is shown in \cref{fig:gadget3}.
	On the path from vertex 8 to vertex 1, waiting time is allowed, but only at vertex 4.
	Symmetrically, there can be no waiting time on the path from 2 to 6 except at vertex 5.
	\begin{align*}
		V&=\{1,\mathellipsis,7 \}\\
		E&=\{(1,3),(3,1),(2,3),(3,2)\}
		\cup \{(6,7),(7,6),(6,8),(8,6)\}
		\\ &\cup \{(i,i+1)| i \in \{3,\mathellipsis, 5\}\} 
		\cup \{(i+1,i)| i \in \{3,\mathellipsis, 5\}\}
		\\D_{1,2}&=D_{2,1}=D_{7,8}=D_{8,7}=2
		\\D_{5,7}&=D_{4,1}=2
		\\D_{8,4}&=D_{2,5}=3
		\\D_{8,1}&=D_{2,7}=6
	\end{align*}
	All other values of $D$ are set to $\infty$. Let w.l.o.g. $\lambda(8,6)=1$.
	As no waiting is allowed at vertex $6$, it has a fixed arrival/departure time of $2$.
	This leads to $\lambda(5,6)=\lambda(7,6)=1$, $\lambda(6,5)=\lambda(6,7)=\lambda(6,8)=2$  and combined with $D_{8,4}=3$ to $\lambda(5,4)=3$.
	On the path from vertex $2$ to vertex $7$, waiting time is only allowed at vertex $5$ and must not exceed one time step.
	Thus, $\lambda(4,5) \in \{3,0\}$ and $\lambda(3,4) \in \{2,3\}$.
	Furthermore the following holds: $\lambda(4,3) \in \{0,1\}$.
	Because there can be no waiting time at vertex 3, it has a fixed arrival/departure time and $\lambda(3,4)-\lambda(4,3) \equiv 1 \pmod{\Delta}$.
	Therefore only $\lambda(4,3) =1$ and $\lambda(3,4) =2$ is feasible.
	Hence, $\lambda(4,5) = 3 = \lambda(5,4)$.
\end{proof}

\subsection{NP-Complete Cases with Quadratic Size Gadgets}
For the remaining cases with odd period $\Delta$, we have found a gadget of quadratic size $2 \cdot (1+\Delta \cdot (k+1))$, the shape of which is reminiscent of a comb.

\begin{restatable}{lem}{oddcomb}\label{thm:GadgetKu}
	There is a gadget $G_i$ that even with the limitation $D_{u,v}\geq \hat{d}(u,v)+k$ enforces $\lambda(v_1,v_2)=\lambda(v_2,v_1)$ for some pair of vertices $(v_1,v_2)$  as long as $\Delta \geq k +2 \land k\geq 1$ and period $\Delta$ is odd.
\end{restatable}
The gadget for $\Delta=3$, $k=1$ is shown in \cref{fig:GadgetUEx}. 

\begin{figure}[bt]
	\begin{minipage}[t]{1\textwidth}
		\centering
		\includegraphics[scale=0.7]{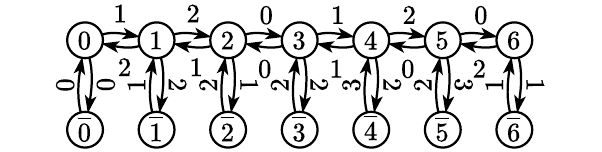}
		\caption{The gadget for $\Delta=3$, $k=1$ which enforces $\lambda(0,\overline{0})=\lambda(\overline{0},0)$, $\lambda(3,\overline{3})=\lambda(\overline{3},3)$ and $\lambda(6,\overline{6})=\lambda(\overline{6},6)$.}
		\label{fig:GadgetUEx}
	\end{minipage}
\end{figure}

\begin{proof}
	
	We show the result for $k=\Delta-2$. 
	This inherits to all smaller $k$ since the limitation for those is weaker and thus complements to \cref{thm:Always} for odd $\Delta$.
	
	The following gadget enforces $\lambda(0,\overline{0})=\lambda(\overline{0},0)$ as shown in \cref{fig:GadgetUkw}:
	\begin{align*}
		V&=\{0,\mathellipsis,(k+1)\Delta, \overline{0},\mathellipsis,\overline{(k+1)\Delta}  \}\\
		E&=\{(i-1,i),(i,i-1)\mid 1\le i\le (k+1)\Delta\}
		\\ &\cup \{(i,\overline{i}),(\overline{i},i)\mid 0\le i\le (k+1)\Delta\}
		\\D_{\overline{i},\overline{j}}&=\hat{d}(\overline{i},\overline{j})+k\mbox{ for all } 0\le i,j \le (k+1)\Delta
	\end{align*}
	
	\begin{figure}[b]
		\begin{minipage}[b]{1\textwidth}
			\includegraphics[width=\linewidth]{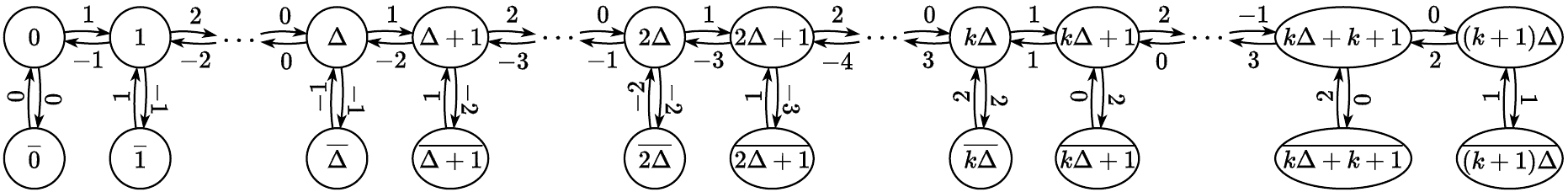}
			\caption{The gadget $(G=(V,E),D, \Delta)$ with a feasible labeling for any odd value of $\Delta$.}
			\label{fig:GadgetUkw}
		\end{minipage}
	\end{figure}
	All other values of $D$ are set to $\infty$. In the following,  we will call the part that consists of the vertices $\{0,\mathellipsis,(k+1)\Delta\}$ the main path. 
	We call the direction from $0$ to $(k+1)\Delta$ the \emph{forward direction}.
	The reverse direction is called the \emph{backward direction}.
	
	We investigate the increase of the maximum waiting time in forward direction to (respectively in backward direction from) $i$ defined by
	\begin{align*}
		w_{\lambda}^+(i)&:=\max_{j < i}\{d(\overline{j},i)-\hat{d}(\overline{j},i)\}\mbox{ and analogously}\\
		w_{\lambda}^-(i)&:=\max_{j < i}{\{d(i,\overline{j})-\hat{d}(i,\overline{j})\}}.	
	\end{align*}
	We show that both $w_{\lambda}^+$ and $w_{\lambda}^-$ have to increase after every $\Delta$ edges on the main path.
	This means there is waiting time in one direction on the main path at all vertices $i\cdot \Delta$.
	
	We have $w_{\lambda}^+(0)=w_{\lambda}^-(0)=-\infty$ since there is no $j<i$.
	We can easily choose $\lambda$ such that $w_{\lambda}^+(1)=w_{\lambda}^-(1)=0$.
	An equivalent inductive definition is 
	\begin{align*}
		w_{\lambda}^+(i+1)=\max\{&
		w_{\lambda}^+(i)+(\lambda(i,i+1)-\lambda(i-1,i)-1) \bmod \Delta,\\&
		(\lambda(i,i+1)-\lambda(\overline{i},i)-1) \bmod \Delta\} \mbox{ 
			and analogously }\\
		w_{\lambda}^-(i+1)=\max\{&
		w_{\lambda}^-(i)+(\lambda(i,i-1)-\lambda(i+1,i)-1) \bmod \Delta,\\&
		(\lambda(i,\overline{i})-\lambda(i+1,i)-1) \bmod \Delta\}.
	\end{align*}
	For the case that there is no waiting at $i$ on the main path in either direction,
	this is simplified to
	\begin{align*}
		w_{\lambda}^+(i+1)&=\max\{w_{\lambda}^+(i),
		(\lambda(i,i+1)-\lambda(\overline{i},i)-1) \bmod \Delta\} \mbox{ 
			and analogously }\\
		w_{\lambda}^-(i+1)&=\max\{w_{\lambda}^-(i),
		(\lambda(i,\overline{i})-\lambda(i+1,i)-1) \bmod \Delta\}.
	\end{align*}
	
	\begin{figure}[t]
		
		\begin{minipage}[t]{0.45\textwidth}
			\centering
			\includegraphics[scale=0.65]{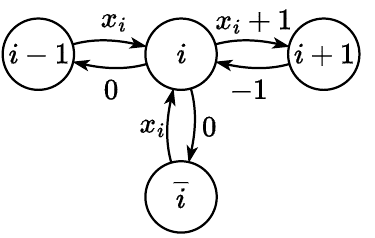}
			\caption{Let w.l.o.g.  $\lambda(i,i-1)=0$. For the case $x_i > w_{\lambda}^+(i)$ we get $w_{\lambda}^-(i+1)=w_{\lambda}^-(i)$.}
			\label{fig:Zinken1}
		\end{minipage}
		\hfill
		\begin{minipage}[t]{0.45\textwidth}
			\centering
			\includegraphics[scale=0.65]{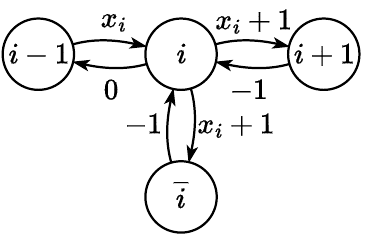}
			\caption{Let w.l.o.g.  $\lambda(i,i-1)=0$. For the case $x_i \le w_{\lambda}^+(i)$ we get $w_{\lambda}^-(i+1)=\max\{w_{\lambda}^-(i),x_i+1\}$.}
			\label{fig:Zinken2}
		\end{minipage}
	\end{figure}
	
	We want to have each increase of the maximum waiting time on the main path as late as possible in forward direction. 
	We first derive conditions under which $ w_{\lambda}^-(i)$ respectively  $w_{\lambda}^+(i)$ must increase. We then show that there is a solution in which we have the increase as late as possible and in which the waiting time just does not exceed $k$.
	At every increase, we have to wait in one direction on the main path.
	As the gadget is symmetrical, an earlier increase is also not possible, as that would mean waiting on the main path at a later point viewed from the other side.
	No increase, that means $w_{\lambda}^-(i+1)=w_{\lambda}^-(i)$, requires that there is no waiting at vertex $i$ on the main path, i.e. $\lambda(i,i-1)\equiv\lambda(i+1,i)+1 \pmod{\Delta}$ (see \cref{fig:Zinken1} and \cref{fig:Zinken2}).
	However, there can be waiting at vertex $i$ on the path from $i+1$ to $\overline{i}$ as long as it does not exceed the previous maximum waiting time. 
	Note that we can wait at $i$ on the way from $i-1$ to $\overline{i}$ at most $k-w_{\lambda}^+(i)$ times, otherwise we would wait on the way from some $\overline{j}$ to $\overline{i}$ more than $k$ times.
	This means $\lambda(i,\overline{i})-(\lambda(i-1,i)+1)\leq k-w_\lambda^+(i)$.
	Let $x_i=(\lambda(i-1,i)-\lambda(i,i-1)) \bmod \Delta$.
	We will discuss two cases: $ x_i> w_{\lambda}^+(i)$ and $x_i \leq w_{\lambda}^+(i)$.
	
	The first case allows us to set $\lambda(i,\overline{i})=(\lambda(i+1,i)+1)\bmod \Delta=\lambda(i,i-1)$ as in \cref{fig:Zinken1}
	since $(\lambda(i,i-1)-(\lambda(i-1,i)+1))\bmod \Delta= (-x_i-1)\bmod \Delta \le \Delta-2-w_{\lambda}^+(i) = k-w_{\lambda}^+(i)$.
	This means there is no waiting at vertex $i$ on the path from $i+1$ to $\overline{i}$. 
	We could also wastefully set $\lambda(i,\overline{i})$ to any value that agrees with both constraints (waiting from $i-1$ and waiting from $i+1$).
	In any case, there is no increase of the maximum waiting time in backwards direction, i.e. $w_{\lambda}^-(i+1)=w_{\lambda}^-(i)$.
	
	In the other case, the smallest timestamp  we can assign to $(i,\overline{i})$ is $(\lambda(i-1,i) +1) \bmod \Delta$ (see \cref{fig:Zinken2}) which leads to the smallest possible waiting time to $\overline{i}$.
	This enforces $w_{\lambda}^-(i+1)>w_{\lambda}^-(i)$ only if the waiting time  exceeds the previous maximum waiting time, i.e.\
	$(\lambda(i,\overline{i})-(\lambda(i+1,i)+1))\bmod \Delta = (\lambda(i-1,i)-\lambda(i+1,i))  \bmod \Delta =x_i +1 \bmod \Delta> w_{\lambda}^-(i)$.
	Therefore, there is an increase with $x_i \geq w_\lambda^-(i)$.

	This means that with $\lambda(i,i-1)\equiv\lambda(i+1,i)+1 \pmod{\Delta}$ an increase of the waiting time $w_{\lambda}^-(i+1)>w_{\lambda}^-(i)$ is enforced if and only if
	$w_{\lambda}^+(i) \ge x_i  \ge w_{\lambda}^-(i)$
	and analogously $w_{\lambda}^+(i+1)>w_{\lambda}^+(i)$ is enforced if and only if
	$w_{\lambda}^-(i) \ge x_i  \ge w_{\lambda}^+(i)$.
	
	Starting as in \cref{fig:GadgetUkw} with $\lambda(0,\overline{0})=\lambda(\overline{0},0)$ and no waiting at vertex 0,
	we get $w_{\lambda}^-(i) = w_{\lambda}^+(i)=0$ for $1\le i\le \Delta$ (since $x_i=2i \pmod{\Delta}$ assumes all values $< \Delta$ once), and only at vertex $\Delta$,
	we get $\lambda(\Delta-1,\Delta)=\lambda(\Delta,\Delta-1)$ which means $x_{\Delta}=0$ enforcing an increase of the waiting time to $w_{\lambda}^-(\Delta+1)=w_{\lambda}^+(\Delta+1)=1$.
	If we would instead start with $\lambda(0,\overline{0})\not=\lambda(\overline{0},0)$ then we would get $x_i=0$ already for an $i< \Delta$ leading to an earlier increase.
	
	\begin{figure}[bt]
		\begin{minipage}[b]{1\textwidth}
			\centering
			\includegraphics[scale=0.65]{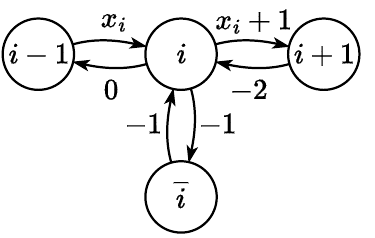}
			\caption{Let w.l.o.g.  $\lambda(i,i-1)=0$. For the case $x_i = w_{\lambda}^-(i)= w_{\lambda}^+(i)$ we choose $\lambda(i+1,i)$ such that $w_{\lambda}^-(i+1)=w_{\lambda}^+(i+1)=w_{\lambda}^-(i)+1$ and in this way the value of $x_{i+1}$ is increased by 3 instead of 2 relative to $x_i$.
				This means that for the following $\Delta-1$ vertices, the value $x_j$ will run through all other values modulo $\Delta$, before it becomes $w_\lambda^-(j)$ again.}
			\label{fig:Zinken3}
		\end{minipage}
	\end{figure}
	
	In fact, we can set $\lambda(\Delta,\Delta+1)=(\lambda(\Delta-1,\Delta)+1)  \bmod \Delta$ and 
	$\lambda(\Delta+1,\Delta)=(\lambda(\Delta,\Delta-1)-2) \bmod \Delta$ as in  \cref{fig:Zinken3}.
	We then set 		$\lambda(\Delta,\overline{\Delta})=\lambda(\overline{\Delta},\Delta)=
	(\lambda(\Delta,\Delta-1)-1)  \bmod \Delta$ and in this way get
	$w_{\lambda}^-(i) = w_{\lambda}^+(i)=1$ for $\Delta < i\le 2\Delta$ and only at vertex $2\Delta$,
	we get $\lambda(2\Delta-1,2\Delta)\equiv\lambda(2\Delta,2\Delta-1)+1 \pmod{\Delta}$ enforcing an increase of the waiting time to $w_{\lambda}^+(2\Delta+1)=w_{\lambda}^+(2\Delta+1)=2$.
	
	By induction on $j$, we can then set $\lambda(j\Delta,j\Delta+1)=(\lambda(j\Delta-1,j\Delta)+1)  \bmod \Delta$ and
	$\lambda(j\Delta+1,j\Delta)=(\lambda(j\Delta,j\Delta-1)-2)  \bmod \Delta$ as well as
	$\lambda(j\Delta,\overline{j\Delta})=\lambda(\overline{j\Delta},j\Delta)=
	(\lambda(j\Delta,j\Delta-1)-j)  \bmod \Delta$.
	This way we get
	$w_{\lambda}^-(i) = w_{\lambda}^+(i)=j$ for $j\Delta < i\le (j+1)\Delta$ and only at vertex $(j+1)\Delta$,
	we get $x_{(j+1)\Delta}=j$ enforcing an increase of the waiting time to $w_{\lambda}^-(j\Delta+1)=w_{\lambda}^+(j\Delta+1)=j$. 
	The increase just reaches vertex $(k+1)\Delta$ at the end of the gadget with the allowed waiting time of $k$, where there is no further increase as it is the last vertex.
	This is accomplished by waiting in one of the two directions on the main path at vertices $i\cdot\Delta$ for $1 \le i \le k$.
	
	However, if we do not wait in one of the two directions on the main path in one of these $k$ cases as in \cref{fig:Zinken3}
	but instead continue as in \cref{fig:Zinken2}, the increase of $w_\lambda^+$ and $w_\lambda^-$ would already be enforced the next time on a position $\frac{\Delta+1}{2}$ later, as shown in \cref{fig:GadgetU0w}.
	In the figure, this position would be $\Delta+\frac{\Delta+1}{2}=\Delta+h+1$.
	Nevertheless, the increase would be enforced after every $\Delta$ steps from this point on, which would lead to a waiting time of $k+1$ before the end of the gadget.
	
\begin{figure}[tb]
	\begin{minipage}[b]{1\textwidth}
		\includegraphics[width=\linewidth]{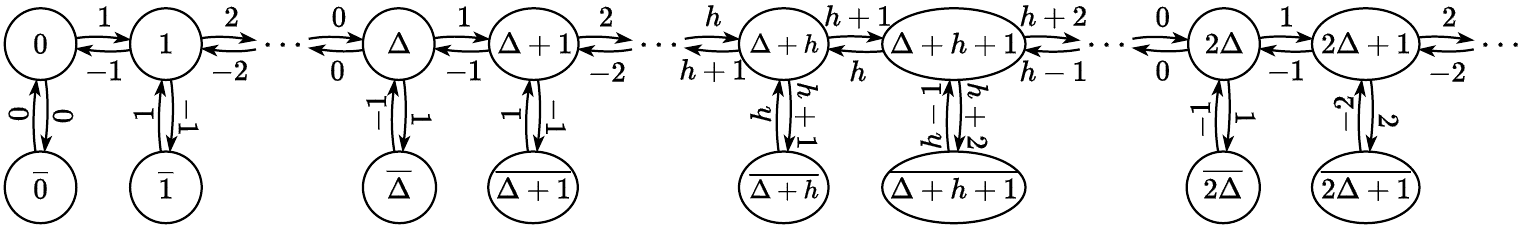}
		\caption{If we do not wait at vertex $\Delta$ in one direction on the main path, we get $w_\lambda^+(\Delta+h+1)=w_\lambda^+(\Delta+h+1)=1$ for $h:=\frac{\Delta-1}{2}$ leading to an increase $w_\lambda^+(\Delta+h+2)=w_\lambda^+(\Delta+h+2)=2$ at vertex $\Delta+h+1$ already.}
		\label{fig:GadgetU0w}
	\end{minipage}
\end{figure}
	
	\begin{figure}[bt]
		\begin{minipage}[b]{1\textwidth}
			\includegraphics[width=\linewidth]{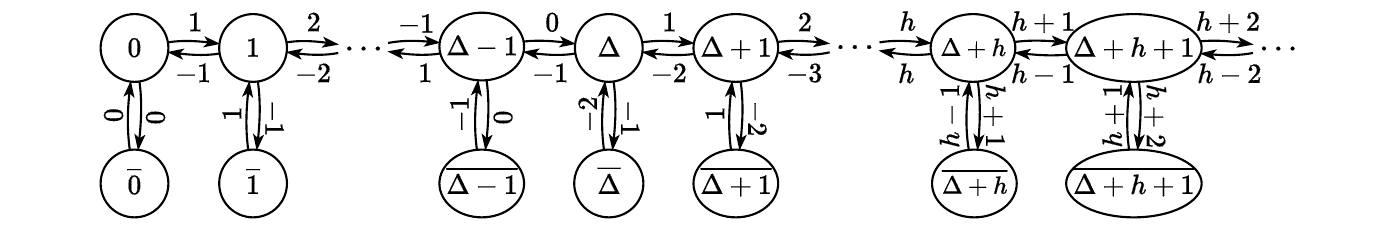}
			\caption{Waiting already at $\Delta-1$ leads to $w_\lambda^+(\Delta)=0$ and $w_\lambda^-(\Delta)=1=x_{\Delta}$ for $h:=\frac{\Delta-1}{2}$ immediately leading to $w_\lambda^-(\Delta+1)=1$
				and even $w_\lambda^+(\Delta+1)=2$
			}
			\label{fig:GadgetUew}
		\end{minipage}
	\end{figure}
	
	Conversely, if we wait earlier than necessary in either direction on the main path,		
	this would mean an increase later than possible in the symmetric case.
	Another way to see this is shown in \cref{fig:GadgetUew}. 
	Waiting early forces an increase at the next multiple of $\Delta$ anyway.
	This means that $\lambda(0,\overline{0})=\lambda(\overline{0},0)$ and symmetrically
	$\lambda((k+1)\Delta,\overline{(k+1)\Delta})=\lambda(\overline{(k+1)\Delta},(k+1)\Delta)$ is enforced in the gadget.
\end{proof}

\begin{figure}[tb]
	\begin{minipage}[b]{1\textwidth}
		\centering
		\includegraphics[scale=0.65]{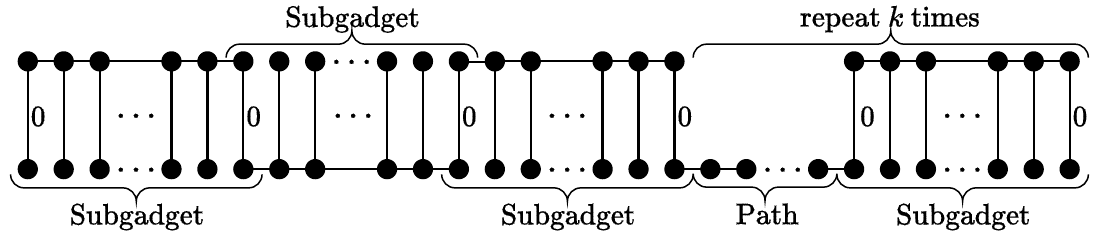}
		\caption{The gadget consists of $\Delta$ subgadgets, where the second is upside down and overlaps wit the first and the third at one edge and the others are connected by a path of length $\Delta-2$. The waiting time between two vertices within a subgadget is limited to $k$.}
		\label{fig:GadgetGkw}
	\end{minipage}
\end{figure}

\begin{figure}[bht]
	\begin{subfigure}[b]{\textwidth}
		\centering
		\includegraphics[width=\linewidth]{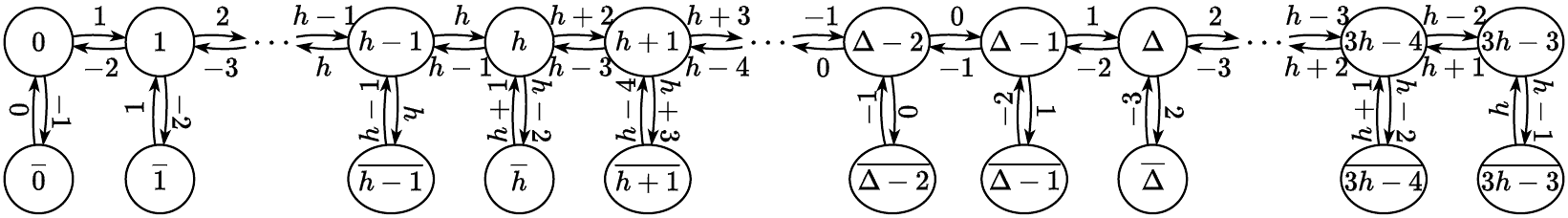}
		\caption{Realization for the gadget from \cref{thm:GadgetKu} for even $\Delta$ starting with $\lambda(0,\overline{0})=\lambda(\overline{0},0)+1$.}
		\label{fig:GadgetGw-a}
	\end{subfigure}
	
	\begin{subfigure}[t]{\textwidth}
		\centering
		\includegraphics[width=\linewidth]{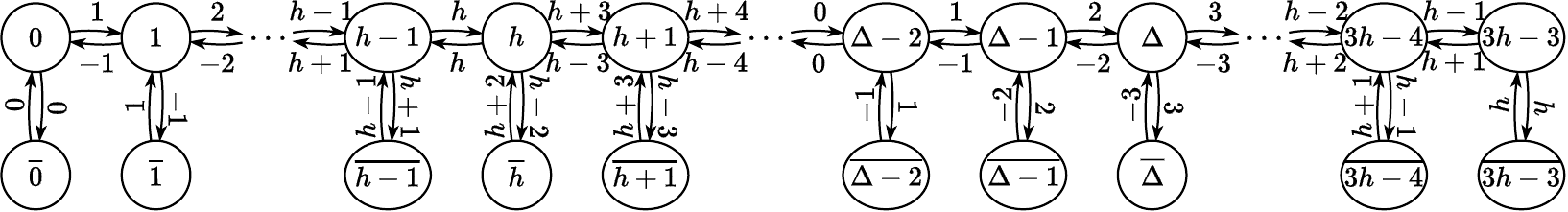}
		\caption{Realization for the gadget from \cref{thm:GadgetKu} for even $\Delta$ starting with $\lambda(0,\overline{0})=\lambda(\overline{0},0)$.}
		\label{fig:GadgetGw-b}
	\end{subfigure}
	\caption{Two realizations for the gadget from \cref{thm:GadgetKu} for even $\Delta$. 
		It turns out that starting the labeling with $\lambda(0,\overline{0})=\lambda(\overline{0},0)+1$ requires waiting one time at some vertex (here $h$) in both directions (see \cref{fig:GadgetGw-a}) but that starting the labeling with $\lambda(0,\overline{0})=\lambda(\overline{0},0)$ requires waiting even two times at some vertex (here $h$) in both directions (see \cref{fig:GadgetGw-b}). 
		Therefore we cannot use this gadget to enforce $\lambda(0,\overline{0})=\lambda(\overline{0},0)$ for even $\Delta$.}
	\label{fig:GadgetGw}
\end{figure}

Finally, we have constructed a similar gadget for the remaining cases with even $\Delta$.

\begin{restatable}{lem}{evencomb}\label{thm:GadgetKg}
	There is a gadget $G_i$ that even with the limitation $D_{u,v}\geq \hat{d}(u,v)+k$ enforces $\lambda(e_1,e_2)=\lambda(e_2,e_1)$ for some pair of vertices $(e_1,e_2)$  as long as $\Delta \geq k +3 \land k\geq 1$ and period~$\Delta$ is even.
\end{restatable}

The gadget is shown in \cref{fig:GadgetGkw}. 	
We show that there cannot be any waiting time in the first subgadget, since the remaining subgadgets and paths already enforce a total waiting time of~$k$ by investigating again the increase in maximum waiting time on the main path to and from some vertex $i$.

\begin{proof}
	For even $\Delta$ the same gadget as in \cref{thm:GadgetKu} does not enforce $\lambda(0,\overline{0})=\lambda(\overline{0},0)$. 
	In fact, starting with $\lambda(0,\overline{0})=\lambda(\overline{0},0)$ requires more waiting than starting with $\lambda(0,\overline{0})\neq\lambda(\overline{0},0)$ as can be seen in \cref{fig:GadgetGw}.
	The comparison of both solutions favors a labeling that does not meet our requirements.
	For this reason the following construction uses a more complex gadget as sketched in \cref{fig:GadgetGkw}.
	We show the result for $k=\Delta-3$. 
	This inherits to all smaller $k$ since the limitation for those is weaker and thus complements to \cref{thm:Always} for even $\Delta$.

	Let $h:= \Delta/2$ and $\ell=3\Delta-3+(2\Delta-3)(\Delta-3)$. 
	The following gadget of length $\hat{d}(\overline{0},\overline{\ell})=\ell+2(\Delta-1)$ as depicted in \cref{fig:GadgetGkw} enforces $\lambda(0,\overline{0})=\lambda(\overline{0},0)$:
	\begin{align*}
		V&=\{0,\mathellipsis,\ell,\overline{0},\mathellipsis,\overline{\ell}  \}\\
		E&=\{(i-1,i),(i,i-1)&&\mid 1\le i\le \Delta -1\}
		\\ &\cup \{(i,\overline{i}),(\overline{i},i)&&\mid 0\le i\le \Delta -1\}
		\\ &\cup\{(\overline{i-1},\overline{i}),(\overline{i},\overline{i-1})&&\mid \Delta \le i\le 2\Delta-2\}
		\\ &\cup \{(i,\overline{i}),(\overline{i},i)&&\mid \Delta -1\le i\le 2\Delta-2\}
		\\ &\cup \{(i-1,i),(i,i-1)&&\mid 2\Delta-1+j(2\Delta-3)\le i\le 3\Delta-3+j(2\Delta-3),
		\\&&& 0 \le j \le \Delta -3\}
		\\ &\cup \{(i,\overline{i}),(\overline{i},i)&&\mid 2\Delta-2+j(2\Delta-3)\le i\le 3\Delta-3+j(2\Delta-3),
		\\&&& 0 \le j \le \Delta -3
		\\ &\cup \{(\overline{i-1},\overline{i}),(\overline{i},\overline{i-1})&&\mid 3\Delta-2+j(2\Delta-3) \le i\le 4\Delta-5+j(2\Delta-3),
		\\&&& 0 \le j < \Delta -4\}
		\\D_{\overline{i},\overline{j}}&=\hat{d}(\overline{i},\overline{j})+k&&\mbox{ for all } 0\le i,j \le \Delta -1
		\\D_{i,j}&=\hat{d}(i,j)+k&&\mbox{ for all } \Delta -1 \le i,j \le 2\Delta-2
		\\D_{\overline{i},\overline{j}}&=\hat{d}(\overline{i},\overline{j})+k&&\mbox{ for all } 2\Delta-2+j(2\Delta-3)\le i,j \le 3\Delta-3+p(2\Delta-3),
		\\&&& 0 \le p < \Delta -3
		\\D_{\overline{0},\overline{\ell}}&=\hat{d}(\overline{0},\overline{\ell})+k
	\end{align*}
	
	All other values of $D$ are set to $\infty$. For simplicity we give the vertex set as $\{0,\mathellipsis,\ell,\overline{0},\mathellipsis,\overline{\ell}  \}$ instead of explicitly removing the isolated vertices.
	We show that there is a labeling for the gadget in which the waiting on the main path takes place exclusively in both directions of the $k$ connecting paths instead of on the subgadgets corresponding to the canonical gadget in \cref{thm:GadgetKu}.
	Furthermore any $\lambda$ has to wait at least two times in any direction in the context (the path or the neighboring subgadgets) of each connecting path.
	This leaves no waiting time for the first subgadget, which enforces $\lambda(0,\overline{0})=\lambda(\overline{0},0)$.

	First we examine the waiting times for each subgadget individually before we investigate them in the context of the whole gadget.
	In the solution without waiting on the main paths of the subgadgets the labeling is the same for each subgadget.
	\begin{claim}\label{claimEven}
		Each subgadget consisting of a comb of $2\Delta$ vertices enforces either $\lambda(0,\overline{0})=\lambda(\overline{0},0)=\lambda(\Delta-1,\overline{\Delta-1})=\lambda(\overline{\Delta-1},\Delta-1)$ or requires waiting on the main path at least two times in the total of both directions.
	\end{claim}
	\begin{claimproof}
		We define the waiting time $w_\lambda^+$ and $w_\lambda^-$ as before in \cref{thm:GadgetKu} and make  similar conclusions which differ in the following statements because $\Delta-k$ is now 3 instead of 2:
		We still have $\lambda(i,\overline{i})-(\lambda(i-1,i)+1)\leq k-w_\lambda^+$ and now
		discuss the two cases $ x_i> w_{\lambda}^+(i)+1$ and $x_i \leq w_{\lambda}^+(i)+1$.
		Again, the first case allows us to set $\lambda(i,\overline{i})=(\lambda(i+1,i)+1)\bmod \Delta=\lambda(i,i-1)$ (see \cref{fig:Zinken1})
		since $(\lambda(i,i-1)-\lambda(i-1,i)-1)\bmod \Delta= (-x_i-1)\bmod \Delta \le \Delta-3-w_{\lambda}^+(i) = k-w_{\lambda}^+(i)$ with no increase of the waiting time in backwards direction, i.e. $w_{\lambda}^-(i+1)=w_{\lambda}^-(i)$.
		In the other case, the smallest timestamp time we can assign $(i,\overline{i})$ is again $(\lambda(i-1,i) +1) \bmod \Delta$.
		This enforces $w_{\lambda}^-(i+1)>w_{\lambda}^-(i)$ only if the waiting time  exceeds the previous maximum waiting time, which means
		$(\lambda(i,\overline{i})-(\lambda(i+1,i)+1))\bmod \Delta = (\lambda(i-1,i)-\lambda(i+1,i))  \bmod \Delta =(x_i +1) \bmod \Delta> w_{\lambda}^-(i)$ (see \cref{fig:Zinken2}).
		Therefore there is an increase with $x_i \geq w_\lambda^-(i)$.
		
		That means with $\lambda(i,i-1)\equiv\lambda(i+1,i)+1 \pmod{\Delta}$ an increase of the waiting time $w_{\lambda}^-(i+1)>w_{\lambda}^-(i)$ is enforced if and only if
		$w_{\lambda}^+(i)+1 \ge x_i  \ge w_{\lambda}^-(i)$
		and analogously $w_{\lambda}^+(i+1)>w_{\lambda}^+(i)$ is enforced if and only if
		$w_{\lambda}^-(i)+1 \ge x_i  \ge w_{\lambda}^+(i)$.
		
		Starting with the equal labeling $\lambda(0,\overline{0})=\lambda(\overline{0},0)$ and no waiting at vertex 0,
		we get $w_{\lambda}^-(i) = w_{\lambda}^+(i)=0$ for $1\le i\le h:=\Delta/2$ and only here,
		we get $\lambda(h-1,h)=\lambda(h,h-1)=h$ enforcing an increase of the waiting time
		to $w_{\lambda}^-(h+1)=w_{\lambda}^+(h+1)=1$.
		But here again we have $x_{h+1}=\lambda(h,h+1)-\lambda(h+1,h)=w_{\lambda}^+(h+1)$ enforcing an increase of the waiting time to $w_{\lambda}^-(h+2)=w_{\lambda}^+(h+2)=2+1=3$.
		By induction on $j$, we have $x_{h+j}=w_{\lambda}^+(h+j)$ enforcing an increase of the waiting time to $w_{\lambda}^+(h+j+1)=w_{\lambda}^+(h+j+1)=2j+1$.
		
		\begin{figure}[t]
			\begin{minipage}[b]{1\textwidth}
				\centering
				\includegraphics[scale=0.6]{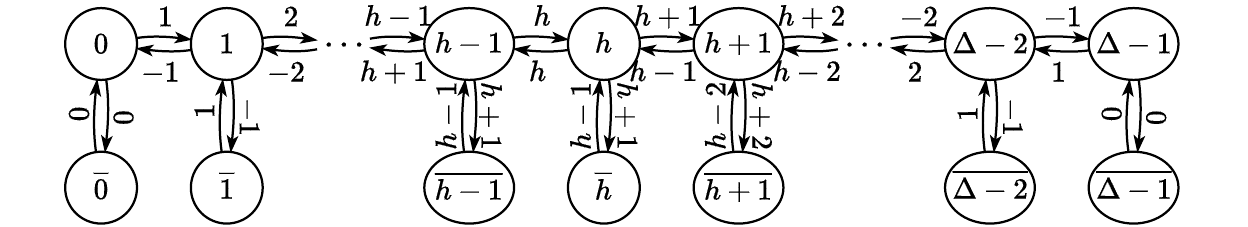}
				\caption{A subgadget with a feasible labeling without waiting on the main path. Observe that $\lambda$ is symmetric on the subgadget.}
				\label{fig:SubgadgetG0w}
			\end{minipage}
		\end{figure}
		
		This just reaches $\Delta-1$ at the end of the subgadget as shown in \cref{fig:SubgadgetG0w} with the allowed waiting time of $k=2h-3$.
		However, waiting on the main path in one of the two directions as in \cref{fig:Zinken3} would not help
		since $x_{h+j}=w_{\lambda}^+(h+j)+1$ still enforces an increase of the waiting time.
		
		The same holds if we start differently as can be seen in the following consideration:
		If we start with values where $\lambda(0,\overline{0})-\lambda(\overline{0},0)$ is even then we would get the latest increase with $\lambda(0,\overline{0})=\lambda(\overline{0},0)$, since $x_i=2i \pmod{\Delta}$ assumes all even values $< \Delta$ once. Values $\lambda(0,\overline{0})\not=\lambda(\overline{0},0)$ would mean a larger difference $x_0$ and would lead to $x_i=0$ already for an $i<h$ leading to an earlier increase. 
		Like before, waiting on the main path in one direction does not change this.
		
		In case $\lambda(0,\overline{0})-\lambda(\overline{0},0)$  is odd we get the latest increase with $x_0=1$, since $x_i=2i+1 \pmod{\Delta}$ assumes all odd values $< \Delta$ once.
		A larger difference $x_0$ would cause
		the increase to happen at latest at vertex $h$, which requires an additional increase at vertex $h+1$ and this in turn an increase at vertex $h+2$ and so on.
		This is shown in \cref{fig:GadgetG0vw}.
		The increase progresses like this until we reach vertex $\Delta-2$ with $x_{\Delta-2}=k$ and $w_\lambda^-(\Delta-2)=w_\lambda^+(\Delta-2)=k-1$, which also enforces an increase.
		We get a maximum waiting time of $k+1$ at vertex $\Delta-1$, therefore not producing a feasible labeling.
		Waiting once in one of the two directions in this case would mean that $\lambda(\overline{\Delta-1},\Delta-1)-\lambda(\Delta-1,\overline{\Delta-1})$ is even. 
		Therefore, we can apply symmetry to show that waiting once still does not help.
	\end{claimproof}
	
	\begin{figure}[tb]
		\begin{minipage}[t]{1\textwidth}
			\centering
			\includegraphics[scale=0.6]{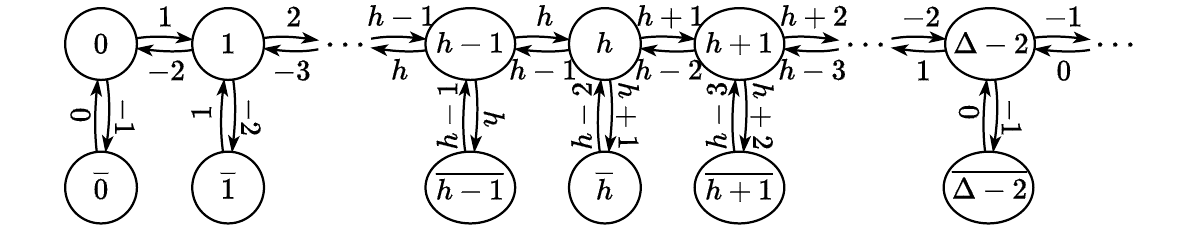}
			\caption{A labeling for the subgadget starting with $x_0$ odd that would lead to $w_\lambda^+(\Delta-1)=k+1$.}
			\label{fig:GadgetG0vw}
		\end{minipage}
	\end{figure}
	
	We can construct a solution by choosing $\lambda$ such that there is no waiting time in any of the subgadgets and for all subgadgets the following holds: $\lambda(0,\overline{0})=\lambda(\overline{0},0)=\lambda(\Delta-1,\overline{\Delta-1})=\lambda(\overline{\Delta-1},\Delta-1)$. 
	As the paths consist of $\Delta-1$ vertices and consequently $\Delta-2$ edges we have to wait one time in each direction on the path for this to be possible.
	As the gadget contains $k$ paths this does not exceed the at most allowed waiting time.
	
	If we do not wait two times on any path we prevent one of the neighboring subgadgets from starting with equal values in both directions on the first or last edge. 
	Because of \cref{claimEven}  we would have to wait at least two times in the respective subgadget.
	Therefore, waiting at least two times is required in the context of each of the $k$ paths and thus, no waiting is possible in the first subgadget.
	Hence, for the first subgadget, which is not contained in the context of any path, $\lambda(0,\overline{0})=\lambda(\overline{0},0)$ holds.
\end{proof}

\begin{theorem}\label{thm:RedToDir}
	For every $\Delta$ and $k$ with $\Delta >k+2$ or $\Delta=k+2$ with $\Delta$ odd, $\Delta$-$k$-\textsc{DiTTR} is NP-complete. Therefore the problem \textsc{DiTTR} is NP-complete.
\end{theorem}
\begin{proof}
	The construction in Proposition 6 of ~\cite{ubTTR_pre} 
	to show that TTR is NP-complete holds for any $\Delta\geq 3$ and all values of $k$ with $k+2\leq \Delta$, because $D_{u,v}\in \{\hat{d}(u,v)+\Delta-2,\infty\}$. Note that the reduction from vertex coloring implies strong NP-completeness.
		This covers all cases as claimed in \cref{fig:hardness}.
		We can combine \cref{lemma:gadgetsRedu} and the lemmas referenced in \cref{fig:hardness} to show that for every $\Delta$ and $k$ with $\Delta >k+2$ or $\Delta=k+2$ with $\Delta$ odd, $\Delta$-$k$-\textsc{DiTTR} is NP-complete.
		That also means \textsc{DiTTR} is NP-complete since already the special versions $\Delta$-$k$-\textsc{DiTTR}, where $k$, which is an implicit parameter of $D$ in the input, and $\Delta$ are two constants in one of the cases (d)-(h) in \cref{fig:hardness}, are NP-complete.
		This also holds if $\Delta$ is part of the input, as \textsc{TTR} is strongly NP-complete.
\end{proof}

{Remark: The proof of \cref{thm:RedToDir} also implies that for even $\Delta>2$ and $\Delta=k+2$ the undirected problem version is NP-complete, while every instance of the corresponding directed problem version is realizable, as indicated in \cref{fig:hardness}.


\section{Conclusion and Further Research}
\label{sec:conclusions}
In this paper, we have initiated the study of the directed version of the graph realization problem for periodic temporal graphs subject to pairwise upper bounds on the fastest paths.
We obtained hardness results for several special cases and identified some easily solvable ones. For trees, we provided a full characterization for all periods $\Delta$ and  all values of the minimum slack parameter
 $k$, giving a lower bound on the maximum allowed waiting time on each path.  

For future work, many problem variants are worth further consideration. An interesting extension would be to also consider upper bounds on the slack. Instead of uniform bounds on the slack, one could also consider multiplicative bounds to reflect that more waiting is acceptable on longer paths. Or one could turn the problem into an optimization problem where one wants to minimize some measure of the deviation from the fastest paths or the desired quality of service.   
In some practical applications, it is useful to further restrict the labeling with additional constraints. For example, when planning train or tram timetables for single-track lines, it is necessary to ensure that a corresponding track section is only served in one direction at a time. Thus, an interesting type of constraint could be to require a solution with $\lambda(a,b)\neq \lambda(b,a)$ for all (or only certain specified) pairs of vertices $a,b\in V$.
According to the proof of \cref{thm:2DTGR_NPC}, this is NP-complete for $\Delta=2$ for general graphs, which motivates the problem for trees but the constructions in this paper using \cref{thm:RedToDir} do not work with this property. 

Further theoretical investigations may consider more general graph classes than just trees.
Finally, it would be interesting to  investigate the practical solvability of instances derived from real network topologies.

\bibliography{Realization}

\begin{thebibliography}{10}

\bibitem{Anstee_1982}
Richard~P. Anstee.
\newblock Properties of a class of (0,1)-matrices covering a given matrix.
\newblock {\em Canadian Journal of Mathematics}, 34(2):438–453, 1982.
\newblock \href {https://doi.org/10.4153/CJM-1982-029-3}
  {\path{doi:10.4153/CJM-1982-029-3}}.

\bibitem{AYOUB1970303}
Jamil~N. Ayoub and Ivan~T. Frisch.
\newblock Degree realization of undirected graphs in reduced form.
\newblock {\em Journal of the Franklin Institute}, 289(4):303--312, 1970.
\newblock \href {https://doi.org/10.1016/0016-0032(70)90273-5}
  {\path{doi:10.1016/0016-0032(70)90273-5}}.

\bibitem{BARNOY2024114810}
Amotz Bar-Noy, David Peleg, Mor Perry, and Dror Rawitz.
\newblock Graph realization of distance sets.
\newblock {\em Theoretical Computer Science}, 1019:114810, 2024.
\newblock \href {https://doi.org/10.1016/j.tcs.2024.114810}
  {\path{doi:10.1016/j.tcs.2024.114810}}.

\bibitem{CHEN1966406}
Wai-Kai Chen.
\newblock On the realization of a (p,s)-digraph with prescribed degrees.
\newblock {\em Journal of the Franklin Institute}, 281(5):406--422, 1966.
\newblock \href {https://doi.org/10.1016/0016-0032(66)90301-2}
  {\path{doi:10.1016/0016-0032(66)90301-2}}.

\bibitem{Edmonds1964}
Jack Edmonds.
\newblock Existence of $k$-edge connected ordinary graphs with prescribed
  degrees.
\newblock {\em J. Res. Nat. Bur. Standards Sect. B}, 68:73--74, 1964.

\bibitem{erdos1960graphs}
Paul Erd{\H{o}}s and Tibor Gallai.
\newblock Graphs with prescribed degrees of vertices.
\newblock {\em Mat. Lapok}, 11:264--274, 1960.

\bibitem{paraTGR_sand}
Thomas Erlebach, Nils Morawietz, and Petra Wolf.
\newblock {Parameterized Algorithms for Multi-Label Periodic Temporal Graph
  Realization}.
\newblock In Arnaud Casteigts and Fabian Kuhn, editors, {\em 3rd Symposium on
  Algorithmic Foundations of Dynamic Networks (SAND 2024)}, volume 292 of {\em
  Leibniz International Proceedings in Informatics (LIPIcs)}, pages
  12:1--12:16, Dagstuhl, Germany, 2024. Schloss Dagstuhl -- Leibniz-Zentrum
  f{\"u}r Informatik.
\newblock URL:
  \url{https://drops.dagstuhl.de/entities/document/10.4230/LIPIcs.SAND.2024.12},
  \href {https://doi.org/10.4230/LIPIcs.SAND.2024.12}
  {\path{doi:10.4230/LIPIcs.SAND.2024.12}}.

\bibitem{fulkerson1960zero}
D.~Ray Fulkerson.
\newblock Zero-one matrices with zero trace.
\newblock {\em Pacific Journal of Mathematics}, 10(3):831--836, 1960.

\bibitem{Hakimi1965DistanceMO}
S.~Louis Hakimi and Stephen~S. Yau.
\newblock Distance matrix of a graph and its realizability.
\newblock {\em Quarterly of Applied Mathematics}, 22:305--317, 1965.

\bibitem{Havel1955}
Václav Havel.
\newblock Poznámka o existenci konečných grafů.
\newblock {\em Časopis pro pěstování matematiky}, 080(4):477--480, 1955.
\newblock URL: \url{http://eudml.org/doc/19050}.

\bibitem{KleitmanWang1976}
Daniel~J. Kleitman and D.~L. Wang.
\newblock Decomposition of a graph realizing a degree sequence into disjoint
  spanning trees.
\newblock {\em SIAM Journal on Applied Mathematics}, 30(2):206--221, 1976.

\bibitem{TGR_arxive}
Nina Klobas, George~B. Mertzios, Hendrik Molter, and Paul~G. Spirakis.
\newblock Realizing temporal graphs from fastest travel times, 2024.
\newblock URL: \url{https://arxiv.org/abs/2302.08860}, \href
  {https://arxiv.org/abs/2302.08860} {\path{arXiv:2302.08860}}.

\bibitem{TGR_sand}
Nina Klobas, George~B. Mertzios, Hendrik Molter, and Paul~G. Spirakis.
\newblock {Temporal Graph Realization from Fastest Paths}.
\newblock In Arnaud Casteigts and Fabian Kuhn, editors, {\em 3rd Symposium on
  Algorithmic Foundations of Dynamic Networks (SAND 2024)}, volume 292 of {\em
  Leibniz International Proceedings in Informatics (LIPIcs)}, pages
  16:1--16:18, Dagstuhl, Germany, 2024. Schloss Dagstuhl -- Leibniz-Zentrum
  f{\"u}r Informatik.
\newblock URL:
  \url{https://drops.dagstuhl.de/entities/document/10.4230/LIPIcs.SAND.2024.16},
  \href {https://doi.org/10.4230/LIPIcs.SAND.2024.16}
  {\path{doi:10.4230/LIPIcs.SAND.2024.16}}.

\bibitem{Lesniak1975}
Linda Lesniak.
\newblock Eccentric sequences in graphs.
\newblock {\em Periodica Mathematica Hungarica}, 6(4):287--293, Dec 1975.
\newblock \href {https://doi.org/10.1007/BF02017925}
  {\path{doi:10.1007/BF02017925}}.

\bibitem{LindnerReisch2022}
Niels Lindner and Julian Reisch.
\newblock An analysis of the parameterized complexity of periodic timetabling.
\newblock {\em J. of Scheduling}, 25(2):157–176, 2022.
\newblock \href {https://doi.org/10.1007/s10951-021-00719-1}
  {\path{doi:10.1007/s10951-021-00719-1}}.

\bibitem{ubTTR_pre}
George~B. Mertzios, Hendrik Molter, Nils Morawietz, and Paul~G. Spirakis.
\newblock Realizing temporal transportation trees, April 2025.
\newblock Extended abstract to appear in Proceedings of 51st International
  Workshop on Graph-Theoretic Concepts in Computer Science (WG 2025), LNCS,
  Springer.
\newblock \href {https://arxiv.org/abs/2403.18513} {\path{arXiv:2403.18513}},
  \href {https://doi.org/10.48550/arXiv.2403.18513}
  {\path{doi:10.48550/arXiv.2403.18513}}.

\bibitem{Odijk1994}
Michiel~A. Odijk.
\newblock Construction of periodic timetables, part 1: A cutting plane
  algorithm.
\newblock Technical report, Technical Report 94-61, TU Delft, 1994.

\bibitem{GBV-196749549}
Øystein Ore.
\newblock {\em Theory of graphs}, volume XXXVIII of {\em American Mathematical
  Society Colloquium Publications}.
\newblock American Mathematical Society, Providence, RI, 1965.
\newblock Second printing.

\bibitem{Peeters2003}
Leon W.~P. Peeters.
\newblock {\em Cyclic railway timetable optimization}.
\newblock PhD thesis, Erasmus Research Institute of Management, Erasmus
  University Rotterdam, The Netherlands, 2003.

\bibitem{RA00}
Klaus Reinhardt and Eric Allender.
\newblock Making nondeterminism unambiguous.
\newblock {\em SIAM Journal on Computing}, 29(4):1118--1131, 2000.
\newblock \href {https://doi.org/10.1137/S0097539798339041}
  {\path{doi:10.1137/S0097539798339041}}.

\bibitem{PESP}
Paolo Serafini and Walter Ukovich.
\newblock A mathematical model for periodic scheduling problems.
\newblock {\em SIAM Journal on Discrete Mathematics}, 2(4):550--581, 1989.
\newblock \href {https://doi.org/10.1137/0402049} {\path{doi:10.1137/0402049}}.

\bibitem{tamura1993realization}
Hiroshi Tamura, Masakazu Sengoku, Shoji Shinoda, and Takeo Abe.
\newblock Realization of a network from the upper and lower bounds of the
  distances (or capacities) between vertices.
\newblock In {\em 1993 IEEE International Symposium on Circuits and Systems
  (ISCAS)}, pages 2545--2548. IEEE, 1993.

\end{thebibliography}
\end{document}